\documentclass{amsart}[11pt]
\usepackage{amsmath, amsfonts, amsthm, amssymb}
\usepackage{subfigure}
\usepackage{verbatim}
\usepackage{hyperref}
\usepackage{color}
\usepackage{enumerate}
\usepackage[dvips]{graphicx}
\numberwithin{equation}{section}
\newtheorem{theorem}{Theorem}[section]
\newtheorem{corollary}[theorem]{Corollary}
\newtheorem{lemma}[theorem]{Lemma}
\newtheorem{proposition}[theorem]{Proposition}

\theoremstyle{definition}

\newtheorem{remark}[theorem]{Remark}
\newtheorem*{notation}{Notation}

\newcommand{\ind}{1\hspace{-2.1mm}{1}}

\newcommand{\D}{\mathrm{d}}
\newcommand{\E}{\mathrm{e}}
\newcommand{\sgn}{\mathrm{sgn}}

\newcommand{\atanh}{\mathrm{arctanh}}
\def\equalDistrib{\,{\buildrel \Delta \over =}\,}

\begin{document}

\title{Large deviations for the extended Heston model: the large-time case}
\author{Antoine Jacquier}
\address{Department of Mathematics, Imperial College London}
\email{ajacquie@imperial.ac.uk}
\author{Aleksandar Mijatovi\'c}
\address{Department of Statistics, University of Warwick}
\email{a.mijatovic@warwick.ac.uk}
\thanks{The authors would like to thank Jim Gatheral and Claude Martini for useful discussions.}
\date{\today}
\begin{abstract}
We study here the large-time behaviour of all continuous affine stochastic volatility models (in the sense of~\cite{KellerRessel})
and deduce a closed-form formula for the large-maturity implied volatility smile.
Based on refinements of the G\"artner-Ellis theorem on the real line, 
our proof reveals pathological behaviours of the asymptotic smile.
In particular, we show that the condition assumed in~\cite{GJ10} 
under which the Heston implied volatility converges to the SVI parameterisation 
is necessary and sufficient. 
\end{abstract}
\maketitle

\section{Introduction}
We are interested here in the large-time behaviour of the process $\left(t^{-1}X_t\right)_{t>0}$,
where $X$ is defined via the system of stochastic differential equations
$$
\left.
\begin{array}{rll}
\D X_t & =\displaystyle -\frac{1}{2}\left(a+V_t\right)\D t+\rho\sqrt{V_t}\,\D W^1_t+\sqrt{a+\left(1-\rho^2\right)V_t}\,\D W^2_t ,
\quad & X_0=x\in\mathbb{R},\\
\D V_t & = \displaystyle \left(b+\beta V_t\right)\D t+\sqrt{\alpha V_t}\,\D W^1_t,
\quad & V_0=v\in(0,\infty),
\end{array}
\right.
$$
with
$a,b\geq 0$, $\alpha>0$, $\beta\in\mathbb{R}$, $\rho\in\left[-1,1\right]$
and $\left(W_t^1,W_t^2\right)_{t\geq 0}$ is a two-dimensional standard Brownian motion.
The couple $(X_t,V_t)_{t\geq 0}$ represents the restriction to continuous paths of the whole class of affine stochastic volatility models with jumps (ASVM), introduced by Keller-Ressel~\cite{KellerRessel}.
In particular it encompasses the popular Heston stochastic volatility model~\cite{Heston},
in which $b>0$ and $\beta<0$.
The weak convergence of the process $\left(t^{-1}X_t\right)_{t>0}$ has been studied 
in~\cite{FJ09,FJM} for the Heston model and 
in~\cite{JKRM} for ASVM, via the G\"artner-Ellis theorem from large deviations theory.
This convergence is the main ingredient needed to obtain the large-maturity behaviour of the implied volatility in these models.
However the authors have imposed technical conditions on the parameters, 
which ensures that the assumptions of the G\"artner-Ellis theorem are met:
(i) the limiting cumulant generating function $\Lambda$ is essentially smooth inside a domain $\mathcal{D}$
and (ii) the interior $\mathcal{D}$ contains the origin.

Even though these conditions are usually satisfied in practice, 
they can actually be broken when calibrating the model for volatile markets. 
In terms of the parameters these two conditions---assumed in~\cite{FJ09,FJM}---read $\beta<0$ and $\beta+\rho\sqrt{\alpha}<0$. 
The second assumption makes sense on equity markets where the correlation is usually negative.
However, on FX markets, the correlation between the asset and its volatility is not necessarily so (see~\cite{Janek} for instance), 
and a large value of the variance of volatility parameter $\alpha$ can violate this assumption.
In~\cite{Andersen}, Andersen and Piterbarg studied the moment explosions of the Heston model (and other stochastic volatility models).
They assume $\beta<0$, but it appears that the restriction $\beta+\rho\sqrt{\alpha}<0$ may also be needed.
In~\cite{Zeliade} the authors highlighted the importance of this latter condition by proving that the Heston model remains of Heston form
under the Share measure (i.e. taking the share price as the numeraire) with new mean-reversion speed $-(\beta+\rho\sqrt{\alpha})$.
This in particular implies that the left wing of the smile could be deduced from the right wing automatically by symmetry. 
This may not be true however when this condition fails.
Reversing the symmetry, the case where the mean-reversion $-\beta$ (in the original measure) is positive becomes interesting to study as well.

We show here that a large deviations principle still holds (as $t$ tends to infinity) for the process
 $\left(t^{-1}X_t\right)_{t>0}$ when the two conditions (i) and (ii) above fail, i.e. without the technical assumptions of~\cite{FJ09, FJM, JKRM}.
As an application, we translate this asymptotic behaviour into asymptotics of the implied volatility, 
corresponding to European vanilla options with payoff $\left(\E^{X_t}-\E^{xt}\right)_+$, 
for any real number $x$.
In~\cite{GJ10}, the authors proved that the so-called \textit{Stochastic Volatility Inspired} (SVI) 
parametric form---first proposed in~\cite{GatheralSVI}---of the implied volatility was the genuine limit 
(as the maturity tends to infinity) of the Heston implied volatility under the same technical conditions as in~\cite{FJ09, FJM, JKRM}. 
We extend the scope of this result by proving that it remains partially true---i.e. on some subsets of the real line---without the technical conditions mentioned above.

In Section~\ref{sec:LDPHeston}, we study the limiting behaviour of the limiting cumulant generating function of the process $\left(t^{-1}X_t\right)_{t>0}$ and state the main result of the paper (Theorem~\ref{thm:LDPAffineThm}), 
i.e. a large deviations principle for this process.
In Section~\ref{sec:RateFunctionsAndOptionPrices}, we translate this LDP into option price and implied volatility asymptotics.
Section~\ref{sec:ProofLDPThm} contains the proof of the main theorem 
and Section~\ref{app:TechLemma} contains some technical results needed in the proof of the main theorem.

\section{LDP for continuous affine stochastic volatility models} \label{sec:LDPHeston}

\subsection{The model and its effective domain}
Throughout this paper we work on a probability space $\left(\Omega,\mathcal{F},\mathbb{P}\right)$ 
equipped with a filtration $\left(\mathcal{F}_t\right)_{t\geq 0}$ supporting two independent Brownian motions $W^1$ and $W^2$.
We consider affine stochastic volatility models in the sense of~\cite{KellerRessel} with continuous paths.
Let $\left(X_t,V_t\right)_{t\geq 0}$ be an affine process 
with state-space $\mathbb R\times \mathbb R_+$
which satisfies the following SDE
\begin{equation}\label{eq:DefAffine}
\left.
\begin{array}{rll}
\D X_t & =\displaystyle -\frac{1}{2}\left(a+V_t\right)\D t+\rho\sqrt{V_t}\,\D W^1_t+\sqrt{a+\left(1-\rho^2\right)V_t}\,\D W^2_t ,
\quad & X_0=x\in\mathbb{R},\\
\D V_t & = \displaystyle \left(b+\beta V_t\right)\D t+\sqrt{\alpha V_t}\,\D W^1_t,
\quad & V_0=v\in(0,\infty),
\end{array}
\right.
\end{equation}
where the admissible parameter values are given by
\begin{align}
\label{eq:Params}
a,b\geq 0,\quad \alpha>0,\quad \beta\in\mathbb{R}\quad\text{and}\quad\rho\in\left[-1,1\right].
\end{align}
The process $\left(V_t\right)_{t\geq 0}$ is a square-root diffusion process and 
the Yamada-Watanabe conditions~\cite{KarSh} ensure that 
a unique non-negative strong solution exists.
The share price process $S=(S_t)_{t\geq0}$,
defined by $S_t:=\exp\left(X_t\right)$, is a local martingale
with respect to the filtration $\left(\mathcal{F}_t\right)_{t\geq0}$,
and~\cite[Theorem~2.5]{KellerRessel} implies that $S$ is a true  martingale.
The Heston model~\cite{Heston} with mean-reversion rate
$\kappa$, positive long-time variance level $\theta$, 
volatility of volatility $\sigma$ and correlation $\rho$, is in the class of models
given by the SDE in~\eqref{eq:DefAffine}
(take $a=0$, $b=\kappa\theta>0$, $\beta=-\kappa<0$, $\alpha=\sigma^2$;
the correlation parameter $\rho$ has the same role as in~\eqref{eq:DefAffine}).

\begin{remark}\text{ }
\begin{itemize}
\item[(i)] The class of models defined by~\eqref{eq:DefAffine} 
coincides with the class of affine stochastic volatility models with continuous sample paths.
\item[(ii)] The parameter $a$ adds modelling flexibility.
\item[(iii)] The general form of the instantaneous variance of a continuous affine stochastic volatility 
log-stock process $X$ is given by 
$a+\widetilde{\alpha}V$
for some $\widetilde{\alpha}>0$. 
A simple scaling of the process $V$ in~\eqref{eq:DefAffine}
maps the class of models given by~\eqref{eq:DefAffine} to the general case.
Without loss of generality we therefore assume
$\widetilde{\alpha}=1$. 
\item[(iv)] The process $U=\left(U_t\right)_{t\geq 0}$ defined by $U_t:=a+V_t$ for all $t\geq 0$ follows
the shifted square-root dynamics (see~\cite{Linetsky} for applications of 
the shifted square-root process in pricing theory).
\end{itemize}
\end{remark}

Let us define the cumulant generating function
\footnote{We will use here the terms ``logarithmic moment generating function'' and ``cumulant generating function'' as synonyms.}
$\Lambda_t$ of the random variable $X_t$, where $X_0=0$, by
\begin{equation}\label{eq:DefLaplace}
\Lambda_t\left(u\right):=\log\mathbb{E}\left(\exp\left(u X_t\right)\right),\quad
\text{for any}\quad
u\in\mathbb R,\quad t\geq0,
\end{equation}
as an extended real number in $(-\infty,\infty]$.
The effective domain of $\Lambda_t$ is defined by
$\mathcal{D}_t:=\left\{u\in\mathbb{R}:\Lambda_t\left(u\right)<\infty\right\}$. 
Note that by the H\"older inequality the function $\Lambda_t$ is convex on $\mathcal{D}_t.$
In order to give the structure of $\Lambda_t(u)$ explicitly we need to define 
\begin{equation}\label{eq:DefChi}
\chi\left(u\right):=\beta+u\rho\sqrt{\alpha},
\end{equation}
as well as
\begin{equation}\label{eq:DefGamma}
\gamma\left(u\right)  := 
\left(\chi\left(u\right)^2+\alpha u\left(1-u\right)\right)^{1/2}
\quad\text{and}\quad 
f_t\left(u\right) := 
\cosh\left(\frac{\gamma\left(u\right)t}{2}\right)-\frac{\chi\left(u\right)}
{\gamma\left(u\right)}\sinh\left(\frac{\gamma\left(u\right)t}{2}\right).
\end{equation}

In Proposition~\ref{prop:LaplaceAffine} 
we show how to express the cumulant generating function 
of $X$ in terms of the logarithmic moment generating function of
model~\eqref{eq:DefAffine} with $a=0$.

\begin{proposition}\label{prop:LaplaceAffine}
The logarithmic moment generating function $\Lambda_t$ 
defined in~\eqref{eq:DefLaplace} reads
$$\Lambda_t\left(u\right)=\Lambda^H_t\left(u\right)+\frac{a}{2}u\left(u-1\right)t,
\quad\text{for all }t\geq 0\text{ and } u\in\mathcal{D}_t,$$
where $\Lambda^H_t$ is given by~\eqref{eq:DefLaplace}
for the process $X$ in~\eqref{eq:DefAffine} with $a=0$.
Furthermore we have
$$\mathcal D_t=\{u\in\mathbb R\>:\>\Lambda_t^H(u)<\infty\}$$
and the following formula holds
\begin{equation}\label{eq:LambdaT}
\Lambda_t^H\left(u\right)=-\frac{2b}{\alpha}\left(
\frac{\chi\left(u\right)t}{2}+\log f_t\left(u\right)\right)+
\frac{u\left(u-1\right)}{f_t\left(u\right)\gamma(u)}\sinh\left(\frac{\gamma\left(u\right)t}{2}\right)v,
\qquad\text{for all}\quad u\in\mathcal D_t.
\end{equation}
\end{proposition}

\begin{proof}

It is well known that the logarithmic moment generating function of an affine process $X$
given as a solution of SDE~\eqref{eq:DefAffine} is of the form
$$\Lambda_t\left(u\right)=\phi_t\left(u\right)+\psi_t\left(u\right)v\quad
\text{for all }t\geq 0\text{ and } u\in\mathcal{D}_t,$$
where the functions 
$\phi_t,\psi_t:\mathcal D_t\to\mathbb R$ 
satisfy the system of Riccati equations (see e.g.~\cite{KellerRessel})
\begin{equation}\label{eq:Riccati}
\left.
\begin{array}{rl}
\partial_t\phi_t\left(u\right) & = F\left(u,\psi_{t}\left(u\right)\right), \qquad \phi_0\left(u\right)=0,\\
\partial_t\psi_t\left(u\right) & = R\left(u,\psi_{t}\left(u\right)\right), \qquad \psi_0\left(u\right)=0,
\end{array}
\right.
\end{equation}
with
$$R\left(u,w\right):=\frac{1}{2}u\left(u-1\right)+\frac{\alpha}{2}w^2+u w \rho\sqrt{\alpha}+\beta w
\qquad\text{and}\qquad
F\left(u,w\right):=\frac{a}{2}u\left(u-1\right)+bw.
$$
The Riccati equation equation for $\psi_t$ can be solved in closed form
$$\psi_t\left(u\right)=\sinh\left(\frac{\gamma\left(u\right)t}{2}\right)
\frac{u\left(u-1\right)}{\gamma\left(u\right)f_t\left(u\right)},$$
where the functions $\gamma$ and $f_t$ are defined in~\eqref{eq:DefGamma}.
The function $\phi_t$ can be determined by noting that equation~\eqref{eq:Riccati} is equivalent to
$\phi_t\left(u\right)=\int_{0}^{t}F\left(u,\psi_s\left(u\right)\right)\D s$.
Therefore 
$\phi_t\left(u\right)=au\left(u-1\right)t/2+b\int_{0}^{t}\psi_s\left(u\right)\D s$.
The function $\Lambda_t^H$ can be constructed in an analogous way on the set
$\{u\in\mathbb R\>:\>\Lambda_t^H(u)<\infty\}$
with $R$ and $F$ as above and $a=0$.
This concludes the proof.
\end{proof}

In order to analyse the effective domain $\mathcal D_t$ we need to introduce the quantities $u_-$ and $u_+$ given by
\begin{equation}\label{eq:DefUM}
u_-:=
\left\{
\begin{array}{ll}
\displaystyle \frac{1}{2\sqrt{\alpha}}\frac{2\beta\rho+\sqrt{\alpha}-\sqrt{\left(2\beta\rho+\sqrt{\alpha}\right)^2+4\beta^2\left(1-\rho^2\right)}}{1-\rho^2},
\quad & \text{if }\left|\rho\right|<1,\\
-\infty,
\quad & \text{if }\left|\rho\right|=1\text{ and }2\beta\rho+\sqrt{\alpha}\leq 0,\\
\displaystyle -\beta^2/\left(2\beta\rho\sqrt{\alpha}+\alpha\right),
\quad & \text{if }\left|\rho\right|=1\text{ and }2\beta\rho+\sqrt{\alpha}>0,
\end{array}
\right.
\end{equation}
and
\begin{equation}\label{eq:DefUP}
u_+:=\left\{
\begin{array}{ll}
\displaystyle \frac{1}{2\sqrt{\alpha}}\frac{2\beta\rho+\sqrt{\alpha}+\sqrt{\left(2\beta\rho+\sqrt{\alpha}\right)^2+4\beta^2\left(1-\rho^2\right)}}{1-\rho^2},
\quad & \text{if }\left|\rho\right|<1,\\
\infty,
\quad & \text{if }\left|\rho\right|=1\text{ and }2\beta\rho+\sqrt{\alpha}\geq 0,\\
\displaystyle -\beta^2/\left(2\beta\rho\sqrt{\alpha}+\alpha\right),
\quad & \text{if }\left|\rho\right|=1\text{ and }2\beta\rho+\sqrt{\alpha}<0.
\end{array}
\right.
\end{equation}
Note that the inequalities $u_-\leq0$ and $u_+\geq1$ hold for all admissible values of the parameters and that in the case
$|\rho|<1$ the parabola $\gamma(u)^2$ is strictly positive on the interior of the interval $[u_-,u_+]$ between its distinct zeros. 
In the case $|\rho|=1$ the graph of the function $\gamma(u)^2$ is a line and either $u_-$ or $u_+$ are infinite.
For notational convenience we shall understand the interval 
$\left[x,y\right]\subset\mathbb R$ 
as $[x,\infty)$ if $y=\infty$ and as $(-\infty,y]$ if $x=-\infty$.
Proposition~\ref{prop:DomainHeston} analyses the structure of the effective domain $\mathcal{D}_t$ 
of the function $\Lambda_t$. 

\begin{proposition}
\label{prop:DomainHeston}
The effective domain $\mathcal D_t$ of the cumulant generating function $\Lambda_t$
(defined in~\eqref{eq:DefLaplace}) satisfies $[0,1]\subset\mathcal D_t$
for all $t\geq0$ and any set of admissible parameter values from~\eqref{eq:Params}.
Furthermore the following statements hold.
\begin{itemize}
\item[(i)] If $\chi(0)\leq0$ 
we have:
\begin{itemize}
\item[(a)] if $\chi\left(1\right) \leq 0$ then $\left[u_-,u_+\right]\subset\mathcal{D}_t$ for any $t>0$; 
\item[(b)] if $\chi\left(1\right)>0$ then for all $t$ large enough there exists $\overline{u}(t)\in(1,u_+)$
such that
$$\lim_{t\to\infty}\overline{u}\left(t\right)=1 \quad\text{and}\quad
[u_-,\overline{u}(t))\subset\mathcal{D}_t\subset\left(-\infty,\overline{u}\left(t\right)\right).$$ 
\end{itemize}
\item[(ii)] If $\chi(0)>0$ we have:
\begin{itemize}
\item[(a)] if $\chi\left(1\right)\leq 0$ then for all large $t$ there exists 
$\underline{u}(t)\in(u_-,0)$ such that
$$\lim_{t\to\infty}\underline{u}\left(t\right)=0 \quad\text{and}\quad
(\underline{u}\left(t\right),u_+]\subset\mathcal{D}_t\subset\left(\underline{u}\left(t\right),\infty\right);$$ 
\item[(b)] if $\chi\left(1\right)>0$ then for large $t$ there exist 
$\underline{u}(t)\in(u_-,0)$ and $\overline{u}(t)\in(1,u_+)$ such that 
$$\lim_{t\to\infty}\underline{u}\left(t\right)=0,\quad
\lim_{t\to\infty}\overline{u}\left(t\right)=1\quad\text{and}\quad
\mathcal{D}_t=\left(\underline{u}\left(t\right),\overline{u}\left(t\right)\right).$$
\end{itemize}
\end{itemize}
\end{proposition}

\begin{remark}
\label{rem:SimpleFacts}
The following elementary facts are useful in the proof of Proposition~\ref{prop:DomainHeston}.
\begin{enumerate}
\item[(I)] Note that 
$u_-=-\infty$ and $u_+=\infty$ if and only if the conditions 
$|\rho|=1$ and $\sqrt{\alpha}+2\rho\beta=0$ hold.
\item[(II)] The condition $\chi(1)\neq 0$ implies that $u_+>1$ since $u_+$ is the largest root of the quadratic
$\gamma(u)^2$ in~\eqref{eq:DefGamma}.
In particular in (i)(b) and (ii)(b) of Proposition~\ref{prop:DomainHeston} 
the interval $(1,u_+)$ is not empty.
\item[(III)]  The condition $\chi(0)\neq 0$ implies that $u_-<0$.
In particular in (ii) we have $\chi(0)=\beta>0$ and hence the interval $(u_-,0)$ is not empty.
\item[(IV)] The interval $[0,1]$ is contained in $\mathcal D_t$ for all $t\geq0$ 
since the stock price process $(S_0\exp(X_t))_{t\geq0}$ is a true martingale.
\item[(V)] If $\chi(0)=0$ then $u_-=0$ and $u_+=1/(1-\rho^2)$ for $|\rho|<1$ and $u_+=\infty$ for $|\rho|=1$. 
\end{enumerate}
\end{remark}

\begin{remark}
The variance process $\left(V_t\right)_{t\geq 0}$ in~\eqref{eq:DefAffine} 
is a time-changed squared Bessel process (see~\cite{Shirakawa}):
$$\left(V_{t}\right)_{t\geq 0}\equalDistrib \E^{\beta t}R^{2}_{\delta,\tau_t},$$
where $\tau_t:=\alpha^{4}\left(1-\E^{-\beta t}\right)/\left(4\beta\right)$, 
and $\left(R^{2}_{\delta,t}\right)_{t\geq 0}$ is a squared Bessel process of dimension
$\delta:=4b/\alpha^{4}$, i.e. $\D R_{t}^{2}=2R_{t}\,\D W_{t}+\delta\,\D t$ and $R^{2}_{\delta,0}=0$.
The sign of $\chi(0)=\beta$ changes the convexity of the time-change $\tau_t$.
\end{remark}

\begin{proof}
Proposition~\ref{prop:LaplaceAffine} implies that it is enough to study the effective domain
of the cumulant generating function $\Lambda^H_t$ of the Heston model.
It is clear that the function $f_t$, defined 
in~\eqref{eq:DefGamma}
by
$$f_t\left(u\right)=\cosh\left(\frac{\gamma\left(u\right)t}{2}\right)
-\frac{\chi\left(u\right)}{\gamma\left(u\right)}\sinh\left(\frac{\gamma\left(u\right)t}{2}\right),$$
will play a key role in in understanding the set $\mathcal D_t$.

\noindent \textbf{Case (i):} 
If we can prove that 
\begin{equation}\label{eq:f_positive}
f_t(u)  > 0, \quad \text{for all } u\in\left[u_-,1\right],
\end{equation}
then Proposition~\ref{prop:LaplaceAffine}
implies that $[u_-,1]\subset\mathcal D_t$
since the functions on both sides of~\eqref{eq:LambdaT} can be analytically extended to a neighbourhood of 
$[u_-,1]$ in the complex plane and hence coincide on the interval.

We now prove~\eqref{eq:f_positive}.
It follows from the definition of $\gamma$ in~\eqref{eq:DefGamma} that
$|\chi(u)/\gamma(u)|\leq1$ for all $u\in[0,1]$ and hence~\eqref{eq:f_positive} holds on $[0,1]$.
It is easy to see that $\lim_{u\searrow u_-}\chi(u)\leq0$.
Since $\chi(0)=\beta\leq0$ we have $\chi(u)\leq0$ for all $u\in[u_-,0]$ which implies~\eqref{eq:f_positive}.

In case (i)(a) assume first that $u_+<\infty$.
Then elementary algebra shows that $\chi(u_+)\leq0$.
Therefore  $\chi(u)\leq0$, and hence $f_t(u)>0$, for all $u\in[1,u_+]$.
If $u_+=\infty$ the condition $\chi(1)\leq0$ implies that $\rho=-1$ and therefore $\chi(u)<0$ for all $u\geq1$.
Hence $f_t(u)\in(0,\infty)$ for all $u\in[1,\infty)=[1,u_+]$.
Proposition~\ref{prop:LaplaceAffine} and the analytic continuation argument as above imply $[u_-,u_+]\subset\mathcal D_t$.

Recall that in case (i)(b) we have $u_+>1$ (see Remark~\ref{rem:SimpleFacts}~(II)).
Let $\overline u(t)$ be the smallest solution of the equation $f_t(u)=0$ in the interval $(1,u_+)$.
Note that, since $\gamma$ is strictly positive on the interval $(1,u_+)$, for a fixed $t$ the equation $f_t(u)=0$ can be rewritten as
\begin{equation}\label{eq:DefF}
t=F(u),\qquad\text{where}\quad 
F(u):=\frac{2}{\gamma(u)}\>\atanh\left(\frac{\gamma(u)}{\chi(u)}\right).
\end{equation}
This equation has a solution in $(1,u_+)$ for large  $t$ since the continuous function  $F$ tends to infinity
as $u$ decreases to $1$ (since $\lim_{u\searrow1}\gamma(u)/\chi(u)=1$).
This also implies that the smallest solution $\overline u(t)$ decreases to one.
The functions on both sides of~\eqref{eq:LambdaT} coincide on $[u_-,1]$,
are analytic on some neighbourhood of this interval in the complex plane and the right-hand side in~\eqref{eq:LambdaT}
is real and finite on $[u_-,\overline u(t))$. 
They must therefore also coincide on $[u_-,\overline u(t))$,
which in particular implies $[u_-,\overline u(t))\subset\mathcal D_t$.
Formula~\eqref{eq:LambdaT} implies that $\overline u(t)$ is not an element of $\mathcal D_t$ and the convexity of
$\Lambda_t$ yields that $\mathcal D_t\cap[\overline u(t),\infty)=\emptyset$.

\noindent \textbf{Case (ii):} 
In case (ii)(a) the condition $\chi(1)\leq0$ implies $\rho<0$ and hence $\chi(u)\leq0$ for all $u\in[1,u_+]$.
Therefore $f_t(u)>0$ on $[1,u_+]$ and hence $[0,u_+]\subset\mathcal D_t$.
Let $\underline u(t)$ be the largest solution of the equation $f_t(u)=0$ in the interval $(u_,0)$.
Since $\lim_{u\nearrow0}(\gamma(u)/\chi(u))=1$, an analogous argument as in the proof of (i)(b) shows that 
$\underline u(t)$ is well defined and the limit in the proposition holds.
The proof for the inclusions follows the same steps as in the proof of (i)(b).

In case (ii)(b) we have $\chi(0)=\beta>0$ and $\chi(1)>0$.
Therefore the definition of $\gamma$, given in~\eqref{eq:DefGamma}, implies
$$
\lim_{u\nearrow0}\frac{\gamma(u)}{\chi(u)}=1\qquad\text{and}\qquad
\lim_{u\searrow1}\frac{\gamma(u)}{\chi(u)}=1
$$
and hence, by~\eqref{eq:DefF}, there exist solutions to the equation $f_t(u)=0$ in both intervals $(u_-,0)$ and $(1,u_+)$.
Let $\underline u(t)$ be the largest solution in $(u_-,0)$ and $\overline u(t)$ the smallest solution in $(1,u_+)$.
An analogous argument to the one in the proofs of (i)(b) and (ii)(a) gives the form of $\mathcal D_t$.
\end{proof}

\subsection{Large deviation principles and the G\"artner-Ellis theorem}
We review here the key concepts of large deviations for a family of real random variables $(Z_t)_{t\geq1}$
and state the G\"artner-Ellis theorem (Theorem~\ref{thm:GartnerEllis}).
A general reference for all the concepts in this section is~\cite[Section 2.3]{DemboZeitouni}.

Assume that the cumulant generating function 
$\Lambda_t^Z(u):=\log \mathbb{E}\left(\E^{uZ_t}\right)$
is finite on some neighbourhood of the origin 
and that for every $u\in\mathbb{R}$ the following limit exists as an extended real number 
\begin{equation}\label{eq:LDP_Assumption}
\Lambda(u) := \lim_{t\to \infty}t^{-1}\Lambda_t^Z(ut).
\end{equation}
Let
$\mathcal D_\Lambda:=\{u\in\mathbb R\>:\>|\Lambda(u)|<\infty\}$
be the effective domain of $\Lambda$ and assume that
\begin{equation}
\label{eq:0_in_Domain_Assumption}
 0  \in  \mathcal D_\Lambda^o,
\end{equation}
where
$\mathcal D_\Lambda^o$ is the interior of $\mathcal D_\Lambda$ (in $\mathbb{R}$).
Since $\Lambda_t^Z$ is convex for every $t$ by H\"older's inequality,
the limit $\Lambda$ is also convex and the set $\mathcal D_\Lambda$ is an interval.
Since $\Lambda(0)=0$, convexity implies  that for any $u\in\mathbb{R}$ we have $\Lambda(u)>-\infty$.
The function $\Lambda:\mathbb{R}\to(-\infty,\infty]$ is said essentially smooth if 
(a)~it is differentiable in $\mathcal D^o_\Lambda$ and 
(b)~it satisfies $\lim_{n\to\infty}|\Lambda'(u_n)|=\infty$ for every sequence 
$(u_n)_{n\in\mathbb N}$ in $\mathcal D^o_\Lambda$ that converges to a boundary point of $\mathcal D^o_\Lambda$.
A cumulant generating function  $\Lambda$ which satisfies (b) is called steep.
The Fenchel-Legendre transform $\Lambda^*$ of $\Lambda$ is defined by the formula
\begin{equation}\label{eq:DefFenchelLegendreTransf}
\Lambda^*(x) := \sup\{ux-\Lambda(u)\>:\>u\in\mathbb{R}\},
\quad\text{for all } x\in\mathbb{R}
\end{equation}
with an effective domain $\mathcal D_{\Lambda^*}:=\{x\in\mathbb R\>:\>\Lambda^*(x)<\infty\}$.
Under certain assumptions $\Lambda^*$ is a good rate function, i.e.
is lower semicontinuous (since it is a supremum of continuous functions),
satisfies $\Lambda^*(\mathbb{R})\subset[0,\infty]$ (since $\Lambda(0)=0$) and the level sets
$\{x\>:\>\Lambda^*(x)\leq y\}$ are compact for all $y\geq0$ (see~\cite[Lemma 2.3.9(a)]{DemboZeitouni}).
In general $\Lambda^*$ can be discontinuous and $\mathcal D_{\Lambda^*}$ can be strictly contained in $\mathbb{R}$
(see~\cite[Section 2.3]{DemboZeitouni} for elementary examples of such rate functions).
We say that the family of random variables $(Z_t)_{t\geq1}$ satisfies the large deviations principle (LDP)
with the good rate function $\Lambda^*$ if for every Borel measurable set $B$ in $\mathbb{R}$
the following inequalities hold 
\begin{equation}
\label{eq:DefLDP}
-\inf\{\Lambda^*(x):x\in B^o\}\leq\liminf_{t\to\infty}\frac{1}{t}\log \mathbb{P}\left[Z_t\in B\right]
\leq
\limsup_{t\to\infty}\frac{1}{t}\log \mathbb{P}\left[Z_t\in B\right]\leq-\inf\{\Lambda^*(x):x\in \overline B\},
\end{equation}
where the interior $B^o$ and the closure $\overline B$ of the set $B$
are taken in the topology of $\mathbb{R}$ and $\inf\emptyset =\infty$.
It is clear from definition~\eqref{eq:DefLDP} that if $(Z_t)_{t\geq1}$ satisfies the LDP and 
$\Lambda^*$ is continuous on $\overline B$,
then 
$\lim_{t\to \infty}t\log \mathbb{P}\left[Z_t\in B\right]
=-\inf\{\Lambda^*(x)\>:\>x\in B\}$.
An element $y\in\mathbb{R}$ is an \textit{exposed point} of $\Lambda^*$ if there exists 
$u_y\in\mathbb{R}$ such that 
\begin{equation} \label{eq:ExposedCondition}
yu_y -\Lambda^*(y)>xu_y-\Lambda^*(x)
\qquad\text{for all}\quad x\in\mathbb{R}\backslash\{y\}.
\end{equation}
Intuitively the exposed points are those at which $\Lambda^*$ is strictly convex (e.g. the second derivative is continuous and strictly positive).
The segments over which $\Lambda^*$ is affine are not exposed. 
Note that~\eqref{eq:ExposedCondition} can only hold for $y\in\mathcal D_\Lambda$ and, 
if $\Lambda$ is differentiable in $\mathcal D_\Lambda^o$, than  $u_y$ is the unique solution of $\Lambda'(u)=y$.
We now state the G\"artner-Ellis theorem the proof of which can be found in~\cite[Section 2.3]{DemboZeitouni}.

\begin{theorem}
\label{thm:GartnerEllis}
Let $(Z_t)_{t\geq1}$ be a family of random variables for which the function 
$\Lambda:\mathbb{R}\to(-\infty,\infty]$
in~\eqref{eq:LDP_Assumption} satisfies~\eqref{eq:0_in_Domain_Assumption}.
Let $F$ be a closed and $G$ an open set in $\mathbb{R}$.
Then the following inequalities hold
\begin{align*}
\limsup_{t\to\infty} t^{-1}\mathbb{P}\left[Z_t\in F\right] & \leq -\inf\{\Lambda^*(x)\>:\>x\in F\},\\
\liminf_{t\to\infty} t^{-1}\mathbb{P}\left[Z_t\in G\right] & \geq -\inf\{\Lambda^*(x)\>:\>x\in G\cap\mathcal E\},
\end{align*}
where 
$\mathcal E:=\{y\in\mathbb{R}\>:\>y\text{ satisfies~\eqref{eq:ExposedCondition} with } u_y\in\mathcal D_\Lambda^o\}$.
Furthermore if $\Lambda$ is essentially smooth and lower semicontinuous, 
then the LDP holds for $(Z_t)_{t\geq1}$ with the good rate function $\Lambda^*$.
\end{theorem}

\subsection{LDP in affine stochastic volatility models} 
In this section we analyse the large deviations behaviour of the family of random variables $Z_t:=X_t/t$ for $t\geq 1$.
Corollary~\ref{cor:LimitLambda}---which follows from Propositions~\ref{prop:LaplaceAffine} and~\ref{prop:DomainHeston}---describes 
the properties of the cumulant generating function $\Lambda$ defined in~\eqref{eq:LDP_Assumption}, 
and its Fenchel-Legendre transform $\Lambda^*$ is studied in Proposition~\ref{prop:RateFunction}.
The main result of this section, Theorem~\ref{thm:LDPAffineThm},
states that the family $(Z_t)_{t\geq 1}$ satisfies a large deviations principle with rate function $\Lambda^*$.

\begin{corollary}
\label{cor:LimitLambda}
The limiting cumulant generating function~\eqref{eq:LDP_Assumption} for the family of random variables 
$(X_t/t)_{t\geq1}$,
where $(X_t)_{t\geq0}$ is defined by SDE~\eqref{eq:DefAffine},is given by
$$\Lambda\left(u\right)=-\frac{b}{\alpha}\left(\chi\left(u\right)+\gamma\left(u\right)\right)+\frac{a}{2}u\left(u-1\right)
\quad\text{for all }u\in\mathcal{D}_\Lambda,$$
with the functions $\chi$ and $\gamma$ given in~\eqref{eq:DefChi} and~\eqref{eq:DefGamma} respectively.
The function $\Lambda$ is infinitely differentiable on the interior $\mathcal D_\Lambda^o$ of its effective domain.
The boundary points $u_-$ and $u_+$, defined in~\eqref{eq:DefUM} and~\eqref{eq:DefUP}, 
can be used to describe the effective domain $\mathcal{D}_\Lambda$ as follows.
\begin{itemize}
\item[(i)] If $\chi\left(0\right)\leq 0$ we have:
\begin{itemize}
\item[(a)] if $\chi\left(1\right)\leq 0$ then 
$\mathcal{D}_\Lambda=\left[u_-,u_+\right]$;
\item[(b)] if $\chi\left(1\right)>0$ then
$\mathcal{D}_\Lambda=\left[u_-,1\right]$.
\end{itemize}
\item[(ii)] If $\chi\left(0\right)>0$ we have:
\begin{itemize}
\item[(a)] if $\chi\left(1\right)\leq 0$ then
$\mathcal{D}_\Lambda=\left[0,u_+\right]$;
\item[(b)] if $\chi\left(1\right)>0$ then
$\mathcal{D}_\Lambda=\left[0,1\right]$.
\end{itemize}
\end{itemize}
\end{corollary}

\begin{figure}[ht]
\centering
\subfigure[Case (i)(a)]{\includegraphics[height=20mm, width=30mm]{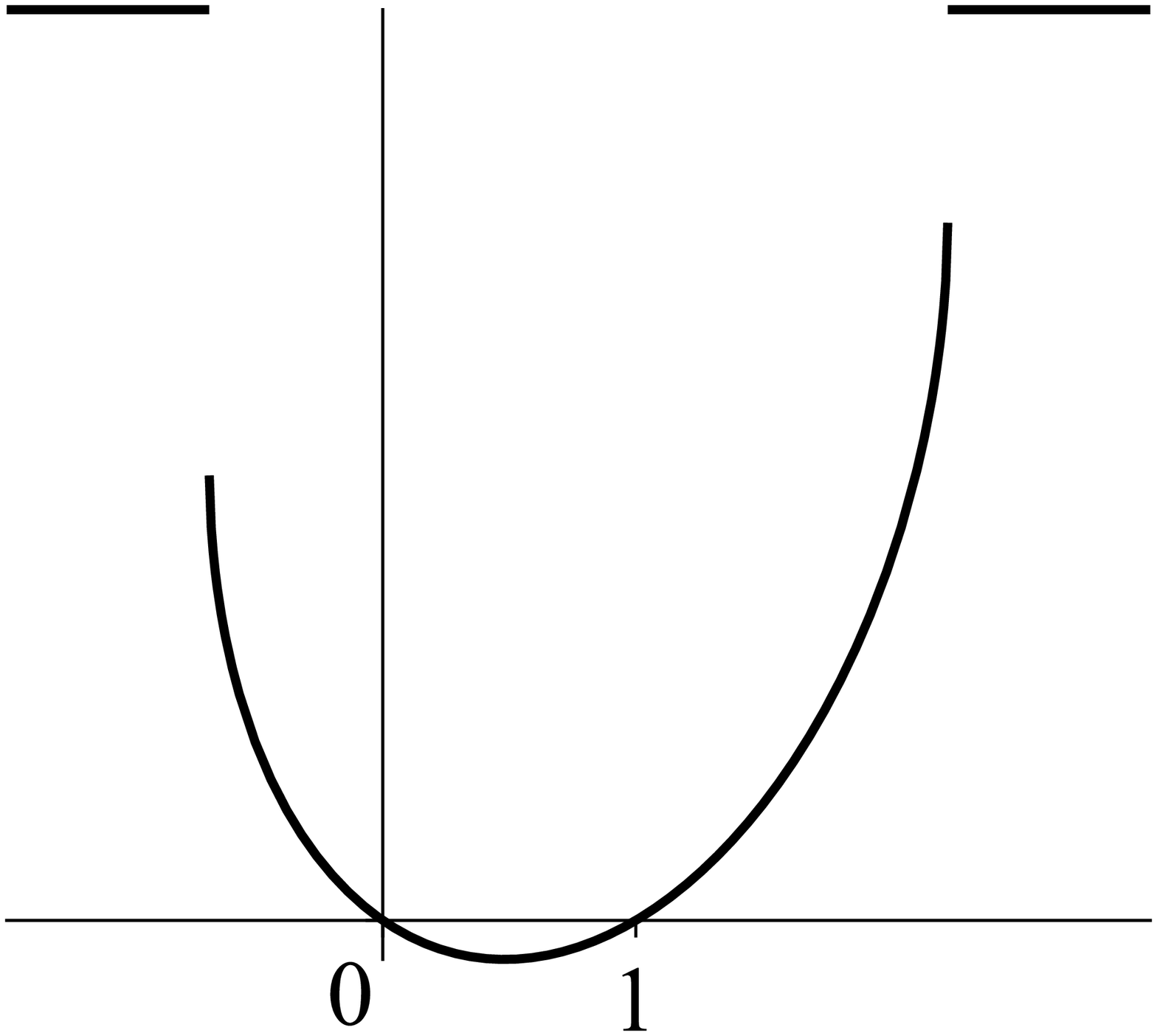}}
\hspace{15pt}
\subfigure[Case (i)(b)]{\includegraphics[height=20mm, width=30mm]{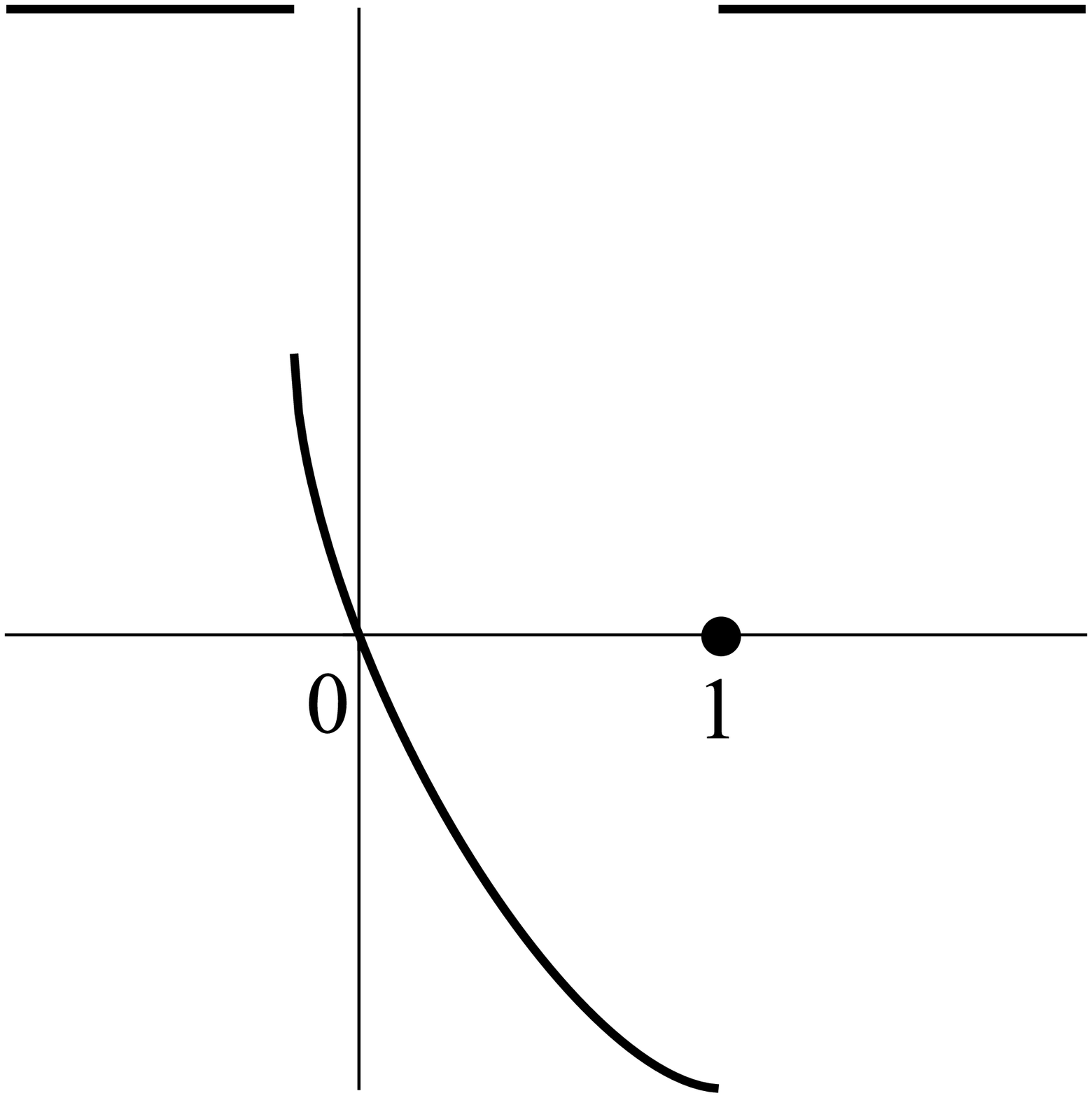}}
\hspace{15pt}
\subfigure[Case (i)(a)]{\includegraphics[height=20mm, width=30mm]{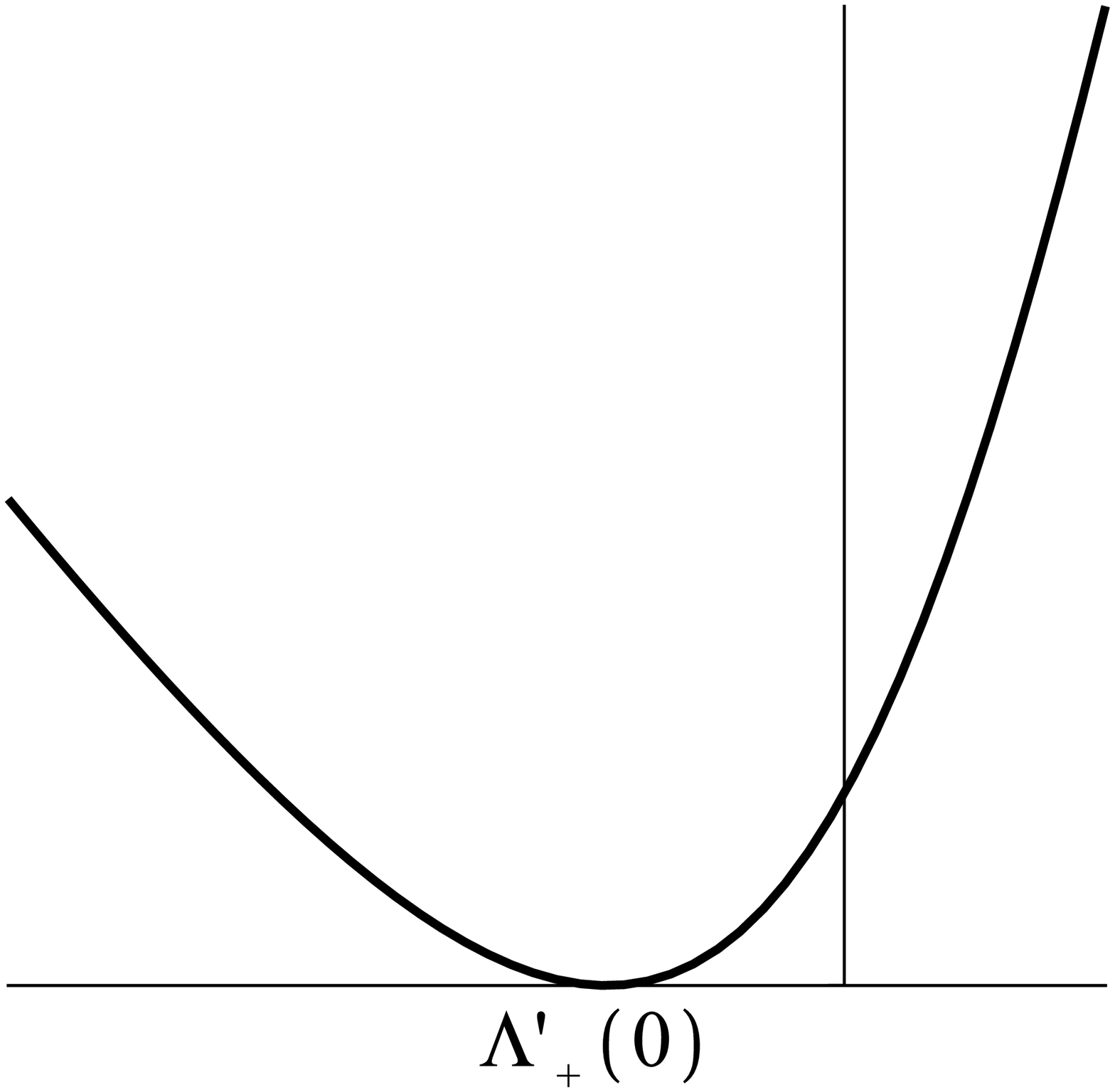}}
\hspace{15pt}
\subfigure[Case (i)(b)]{\includegraphics[height=20mm, width=30mm]{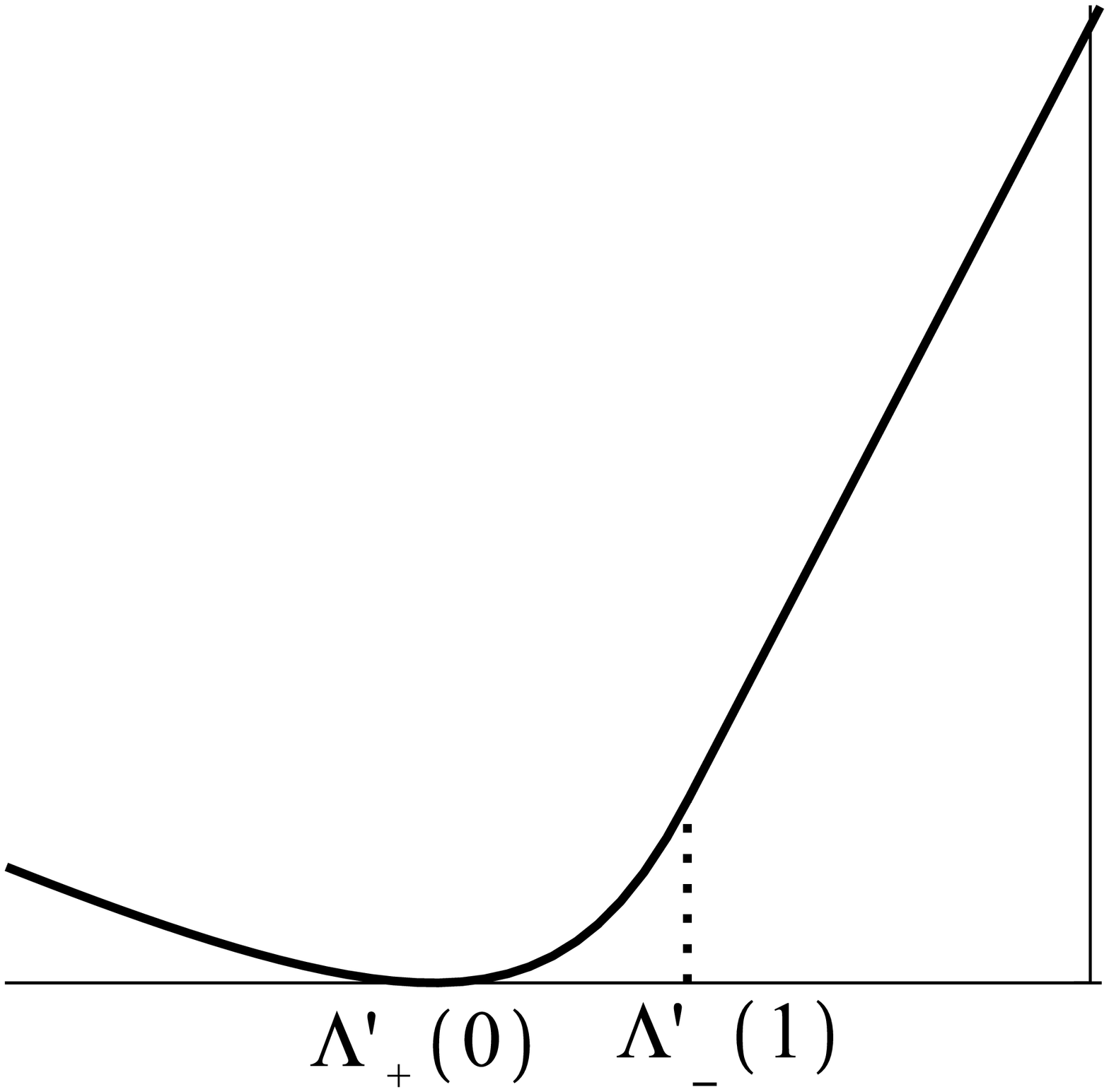}}
\subfigure[Case (ii)(a)]{\includegraphics[height=20mm, width=30mm]{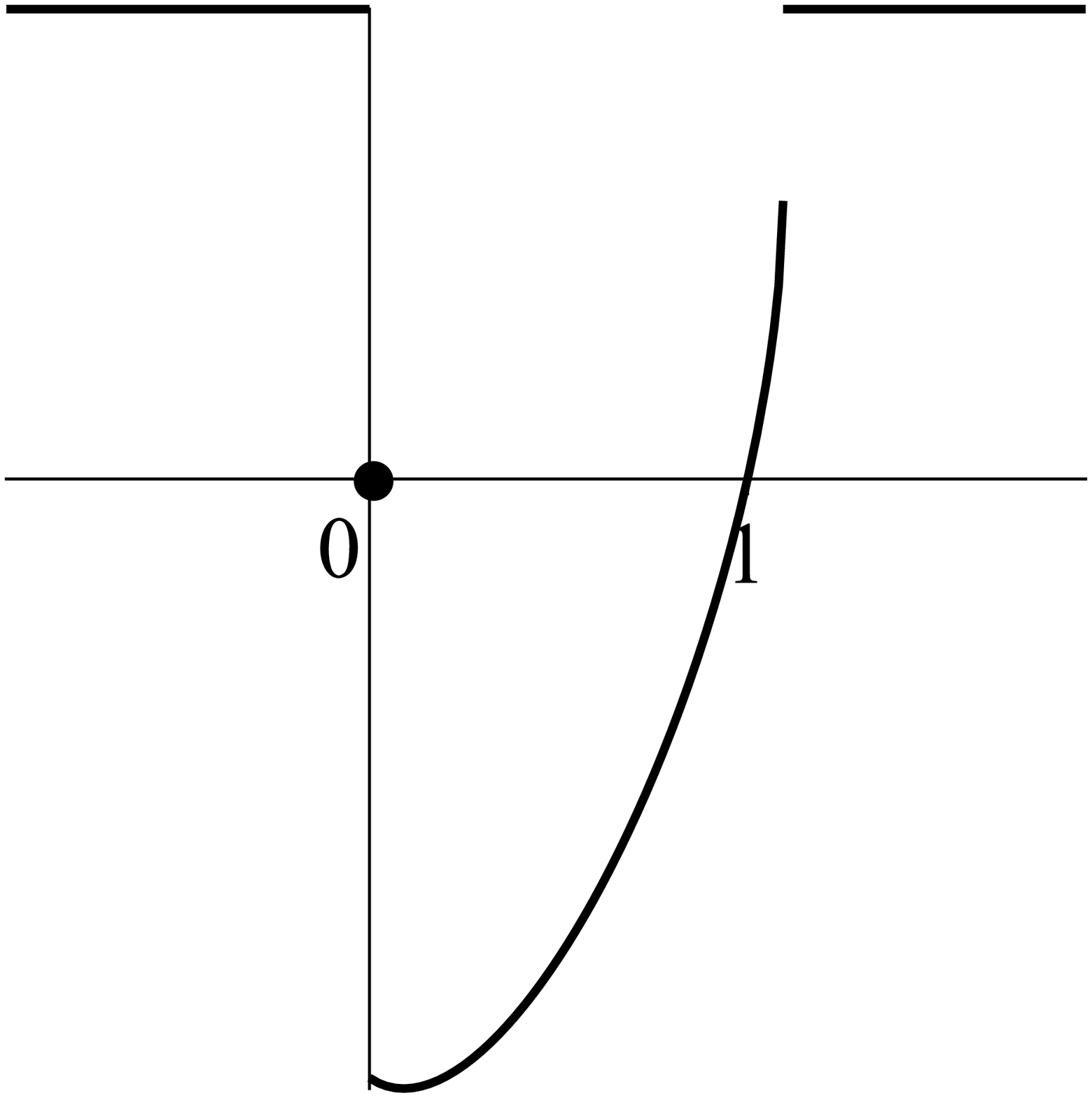}}
\hspace{15pt}
\subfigure[Case (ii)(b)]{\includegraphics[height=20mm, width=30mm]{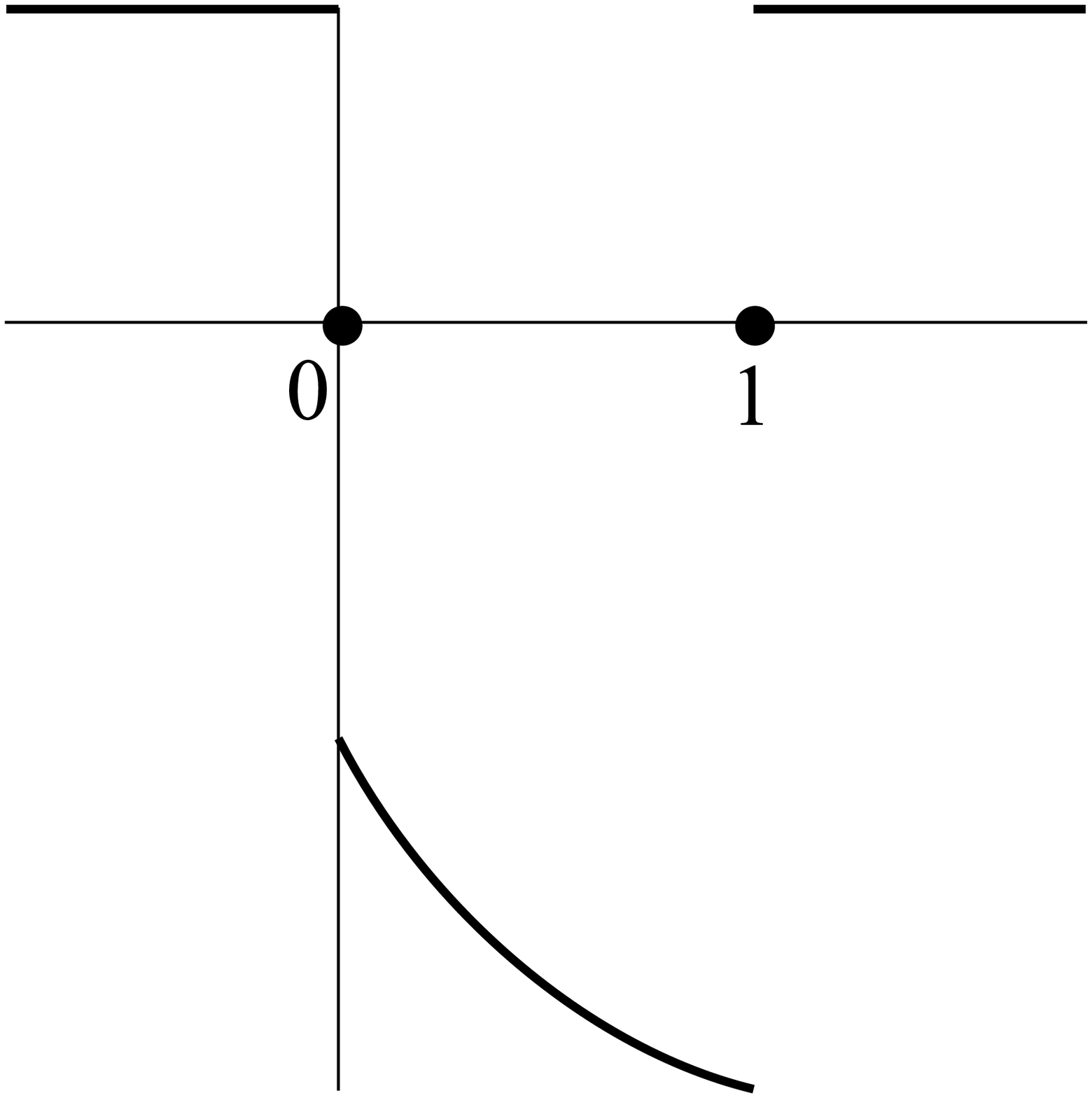}}
\hspace{15pt}
\subfigure[Case (ii)(a)]{\includegraphics[height=20mm, width=30mm]{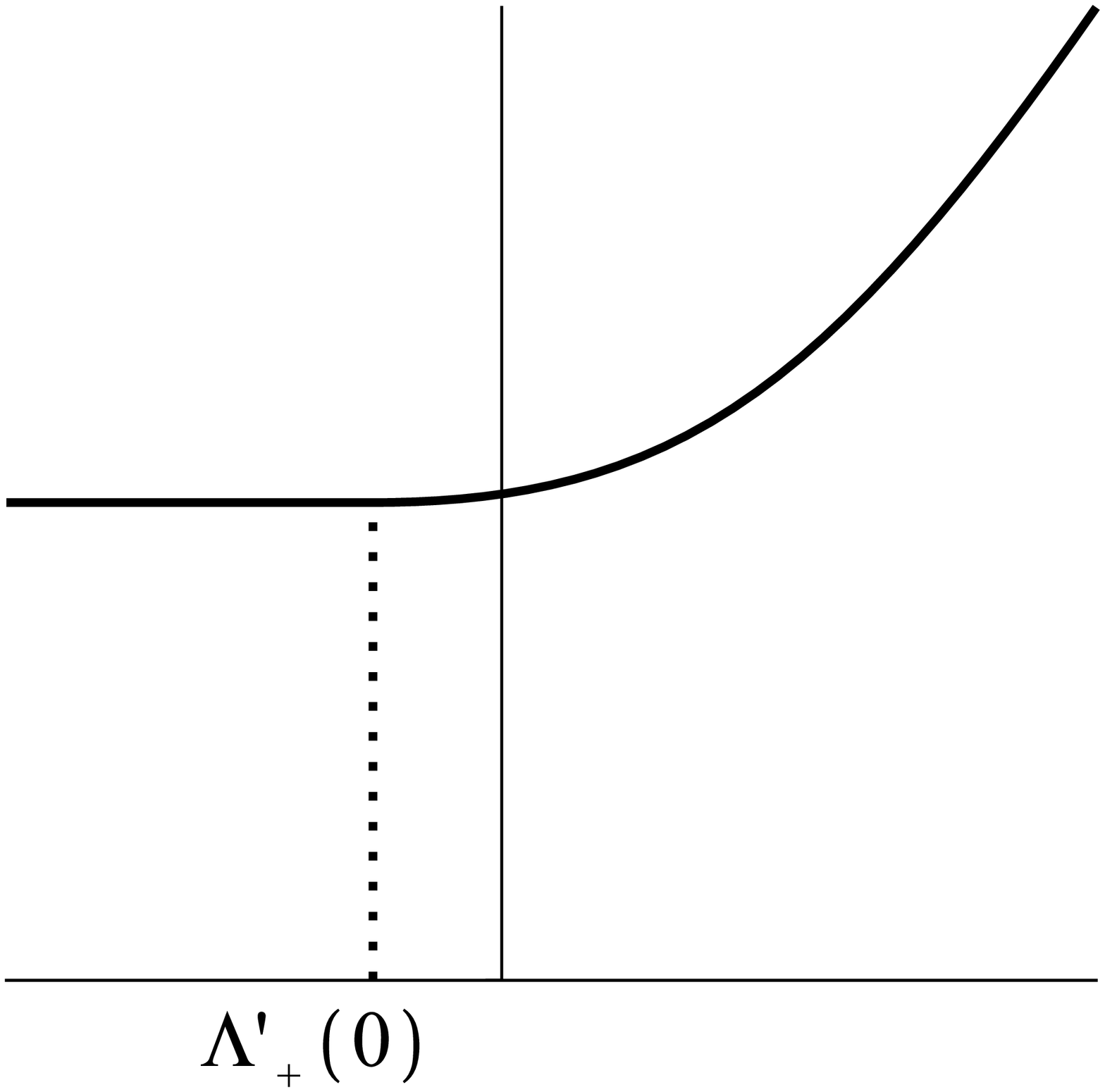}}
\hspace{15pt}
\subfigure[Case (ii)(b)]{\includegraphics[height=20mm, width=30mm]{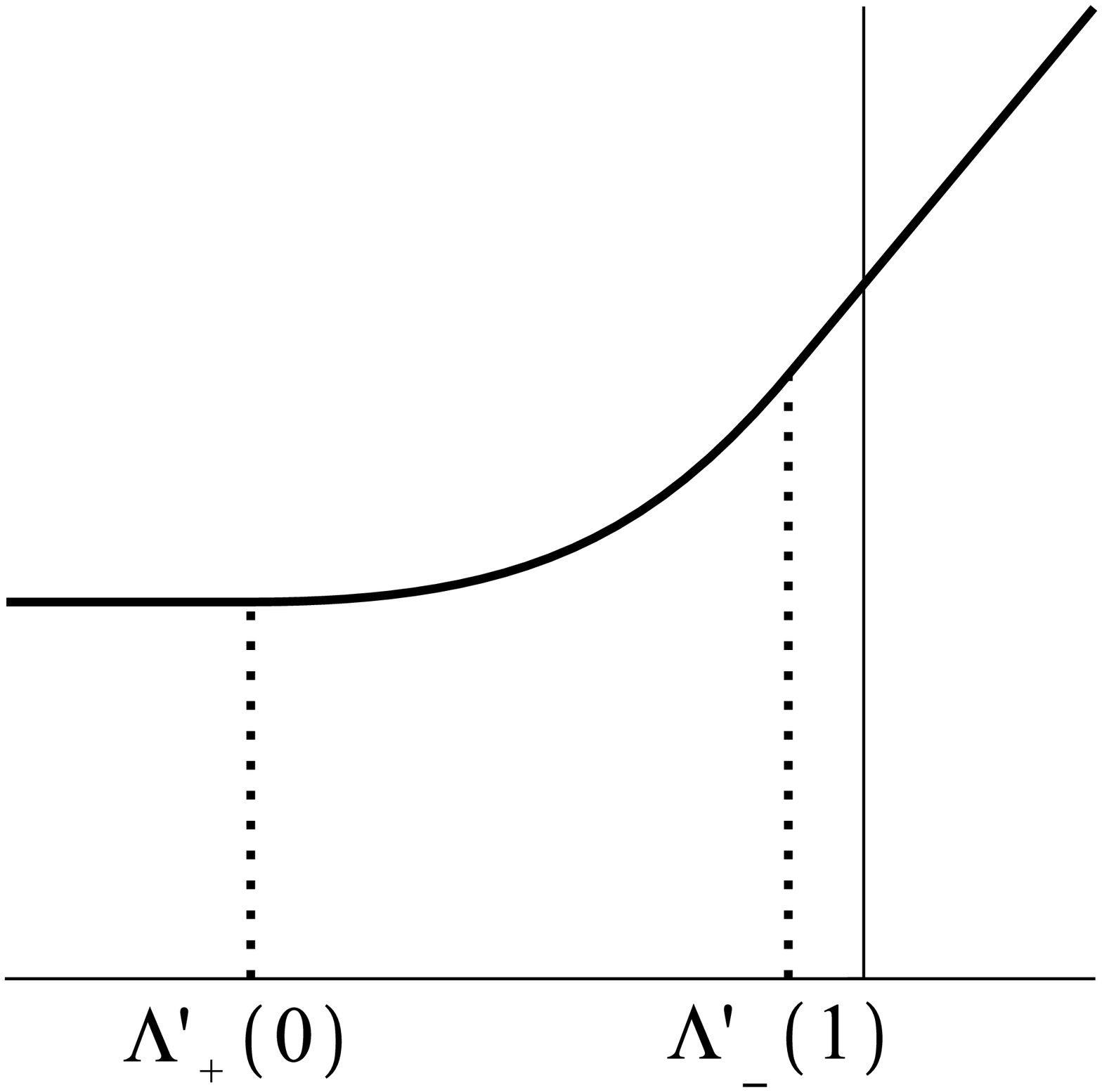}}
\caption{
The four figures on the left represent the function $\Lambda$ characterised in Corollary~\ref{cor:LimitLambda}.
The four figures on the right represent the Fenchel-Legendre $\Lambda^*$ determined in Proposition~\ref{prop:RateFunction}.
The dotted line on the graphs for $\Lambda^*$ represent the threshold $\Lambda'_-\left(1\right)$ and $\Lambda'_+\left(0\right)$ 
above or below which $\Lambda^*$ becomes linear.}
\end{figure}

\begin{remark}\label{rem:LDPLambda}
From Corollary~\ref{cor:LimitLambda}, the following facts can be deduced immediately 
for the large deviations behaviour of the family of random variables $(X_t/t)_{t\geq1}$.
\begin{enumerate} 
\item[(I)] In case (i)(a) the function $\Lambda$ is essentially smooth.
\item[(II)] In case (i)(b) (resp. (ii)(a)) the function $\Lambda$ 
is steep at the left boundary $u_-$ (resp. right boundary $u_+$) but not at the right (resp. left) boundary of the effective domain.
\item[(III)] In case (i)(b) (resp. (ii)(a)) the right (resp. left) boundary point of the effective domain
is strictly smaller (resp. greater) than $u_+$ (resp. $u_-$).
This is a consequence of~(II) and~(III).
\item[(IV)] In case (ii)(b) the function $\Lambda$ is not steep at either of the two boundaries of its effective domain. 
Furthermore $\mathcal D_\Lambda$ is contained in the interior of the interval $[u_-,u_+]$
by~(II) and~(III).
\item[(V)] As a consequence of (I)--(IV) the limiting cumulant generating function $\Lambda$
is steep at a boundary point of the effective domain if and only if this point is an element of the set $\{u_-,u_+\}$.
\end{enumerate} 
\end{remark}

Note that when $u_-$ (resp. $u_+$) is not in $\mathcal{D}_\Lambda$ then the function $\Lambda$ is discontinuous
at $0$ (resp. at $1$). 
We henceforth define the following extended real numbers
\begin{equation}
\label{eq:DefLimits}
\Lambda_-\left(1\right)  :=\lim\limits_{u\nearrow 1}\Lambda\left(u\right),
\quad
 \Lambda_+\left(0\right)  :=\lim\limits_{u\searrow 0}\Lambda\left(u\right), 
 \quad
\Lambda'_-\left(1\right)  :=\lim\limits_{u\nearrow 1}\Lambda'\left(u\right),
\quad
 \Lambda'_+\left(0\right)  :=\lim\limits_{u\searrow 0}\Lambda'\left(u\right).
\end{equation}
The functions $\Lambda$ and $\Lambda'$ are monotone on the intervals $\left(0,\varepsilon\right)$
and $\left(1-\varepsilon,1\right)$ for small enough $\varepsilon$, hence all the limits exist.
Note further that the limit $\Lambda'_+\left(0\right)$ (resp. $\Lambda'_-\left(1\right)$) is equal to $-\infty$ 
(resp. $\infty$) if and only if $\chi\left(0\right)=0$ (resp. $\chi\left(1\right)=0$).

\begin{remark}\label{rem:Values01}
At zero and one the following identities hold
\begin{align*}
\Lambda_+\left(0\right) & = 
-\frac{b}{\alpha}\left(\chi\left(0\right)+\left|\chi\left(0\right)\right|\right)
\quad \text{and}\quad
\Lambda'_+\left(0\right) = 
\left\{
\begin{array}{ll}
\displaystyle
\frac{1}{\left|\chi\left(0\right)\right|}
\left(\left(\chi\left(1\right)-\chi\left(0\right)\right)\Lambda_+\left(0\right)-\frac{b}{2}\right)-\frac{a}{2},
& \text{if } \chi\left(0\right)\ne 0,\\
-a/2,
& \text{if } \chi\left(0\right)=0,\  b=0,\\
-\infty,
& \text{if } \chi\left(0\right)=0,\  b\ne 0,
\end{array}
\right.
\\
\Lambda_-\left(1\right) & = 
-\frac{b}{\alpha}\left(\chi\left(1\right)+\left|\chi\left(1\right)\right|\right)
\quad\text{and}\quad
\Lambda'_-\left(1\right) =
\left\{
\begin{array}{ll}
\displaystyle
\frac{1}{\left|\chi\left(1\right)\right|}
\left(\left(\chi\left(1\right)-\chi\left(0\right)\right)\Lambda_-\left(1\right)+\frac{b}{2}\right)+\frac{a}{2},
& \text{if } \chi\left(1\right)\ne 0,\\
a/2,
& \text{if } \chi\left(1\right)=0,\  b= 0,\\
\infty,
& \text{if } \chi\left(1\right)=0,\  b\ne 0.
\end{array}
\right.
\end{align*}
Note that the inequalities $\Lambda_+\left(0\right)\leq0 $ and $\Lambda_-\left(1\right) \leq 0$ hold for any admissible set of parameters.
The case $\chi(0)=0$ and $b=0$ is rather degenerate, and we refer the reader to Remark~\ref{rem:DegenerateCase}
for further details.
\end{remark}

\begin{proposition}\label{prop:RateFunction}
The Fenchel-Legendre transform $\Lambda^*$ defined in~\eqref{eq:DefFenchelLegendreTransf} 
for the family of random variables $(X_t/t)_{t\geq1}$,
where $(X_t)_{t\geq0}$ is given by SDE~\eqref{eq:DefAffine},
can be represented as follows
\begin{equation}
\label{eq:FormLambdaStar}
\Lambda^*\left(x\right)=
\left\{ \begin{array}{ll}
x u_x-\Lambda\left(u_x\right),
\quad & \text{for all }x\in\Lambda'\left(\mathcal{D}^o_\Lambda\right),\\
x-\Lambda_-\left(1\right),
\quad & \text{for all }
x\in\left[\Lambda'_-\left(1\right),\infty\right)\cap\left(\mathbb{R}\backslash\Lambda'\left(\mathcal{D}^o_\Lambda\right)\right),\\
-\Lambda_+\left(0\right),
\quad & \text{for all }
x\in\left(-\infty,\Lambda'_+\left(0\right)\right]\cap\left(\mathbb{R}\backslash\Lambda'\left(\mathcal{D}^o_\Lambda\right)\right),\\
\end{array}
\right.
\end{equation}
where $u_x$ is the unique solution in $\mathcal{D}^o_\Lambda$ to the equation 
$\Lambda'\left(u\right)=x$ for all $x\in\Lambda'\left(\mathcal{D}^o_\Lambda\right)$.
Furthermore $\Lambda^*$ is continuously differentiable on 
its effective domain $\mathcal{D}_{\Lambda^*}$
and $\mathcal{D}_{\Lambda^*}=\mathbb{R}.$
\begin{itemize}
\item[(i)] The function $\Lambda^*$ attains its global minimal value $-\Lambda_+\left(0\right)$
at $\Lambda'_+(0)$.
If  $0\in\mathcal D_\Lambda^o$ then the minimum is attained at the unique point
$\Lambda'_+(0)=\Lambda'(0)$ and the minimal value is
$\Lambda^*(\Lambda'(0))=\Lambda_+\left(0\right)=0$.
If $0\notin\mathcal D_\Lambda^o$ the minimal value is attained at every
$x\in\left(-\infty,\Lambda'_+\left(0\right)\right]\cap\left(\mathbb{R}\backslash\Lambda'\left(\mathcal{D}^o_\Lambda\right)\right)$
\item[(ii)] The function $x\mapsto\Lambda^*(x)-x$ attains its global minimal value
$-\Lambda_-\left(1\right)$ at $\Lambda'_-(1)$.
If $1\in\mathcal D_\Lambda^o$ then the minimum value 
$\Lambda_-\left(1\right)=\Lambda(1)=0$
is attained at the unique point $\Lambda'_-(1)=\Lambda'(1)$
which is therefore the unique solution to the equation $\Lambda^*(x)=x$.
If $1\notin\mathcal D_\Lambda^o$ the function $x\mapsto\Lambda^*(x)-x$
attains the minimal value at every
$x\in\left[\Lambda'_-\left(1\right)\infty\right)\cap\left(\mathbb{R}\backslash\Lambda'\left(\mathcal{D}^o_\Lambda\right)\right)$.
\end{itemize}
\end{proposition}

\begin{remark}\label{rem:LambdaStar}\text{}
\begin{enumerate}[(i)]
\item Since $\Lambda$ is a strictly convex smooth function on $\mathcal{D}^o_\Lambda$, 
the first derivative $\Lambda'$ is invertible on this interval and 
$u_x$ is a strictly increasing, differentiable function of $x$ on $\Lambda'\left(\mathcal{D}^o_\Lambda\right)$.
Furthermore the equality 
$(\Lambda^*)'\left(x\right)=u_x$ 
holds for any $x\in\Lambda'\left(\mathcal{D}^o_\Lambda\right)$.
\item Corollary~\ref{cor:LimitLambda}
implies the following form for the interval
$\Lambda'(\mathcal D^o_\Lambda)$:
\begin{equation}
\label{eq:FormOfDomain}
\Lambda'\left(\mathcal D^o_\Lambda\right)=
\left\{
\begin{array}{cl}
\mathbb{R},
\quad & \text{if }\chi(0)\leq0,\>\chi(1)\leq0,\\
\left(-\infty,\Lambda'_-(1)\right),
\quad & \text{if }\chi(0)\leq0,\>\chi(1)>0,\\
\left(\Lambda'_+(0),\infty\right),
\quad & \text{if }\chi(0)>0,\>\chi(1)\leq0,\\
\left(\Lambda'_+(0),\Lambda'_-(1)\right),
\quad & \text{if }\chi(0)>0,\>\chi(1)>0.
\end{array}
\right.
\end{equation}
Hence the second case in~\eqref{eq:FormLambdaStar} corresponds to $\chi\left(1\right)>0$ 
and the third case occurs when $\chi\left(0\right)>0$.
\item
When $a$ is null, the unique solution $u_x$ to the equation
$\Lambda'\left(u\right)=x$, when $x\in\Lambda'\left(\mathcal{D}^o_\Lambda\right)$ is given by
\begin{equation}\label{eq:Ux}
u_x=\frac{1}{2\left(1-\rho^2\right)\sqrt{\alpha}}\left(2\rho\beta+\sqrt{\alpha}
+\frac{p\left(x\right)\xi}{\sqrt{p\left(x\right)^2+b^2\left(1-\rho^2\right)}}\right),
\end{equation}
where
$$p\left(x\right):=b\rho+x\sqrt{\alpha},
\quad\text{and}\quad
\xi:=\sqrt{\left(2\rho\beta+\sqrt{\alpha}\right)^2+4\beta^2\left(1-\rho^2\right)}.
$$
This, together with~\eqref{eq:FormLambdaStar}, yields an explicit formula for the rate function $\Lambda^*$.
Note that $u_x$ is well defined as a limit when $\left|\rho\right|$ tends to $1$ and
\begin{equation}\label{eq:upmrho1}
u_{x}=
\frac{1}{4}\frac{b-2\beta x}{2\beta+\rho\sqrt{\alpha}}\frac{4b\beta+\rho(b+2\beta x)\sqrt{\alpha}}{\left(b\rho+x\sqrt{\alpha}\right)^2},
\qquad \text{whenever } \rho\in\left\{-1,1\right\}.
\end{equation}
\item When the parameter $a$ is not null, we do not have a closed-form representation for $u_x$, 
and hence not for the function $\Lambda^*$ either.
However computing $\Lambda^*$ is a simple root-finding exercise and the smoothness of the function $\Lambda$ makes it computationally quick.
\end{enumerate}
\end{remark}

\begin{proof}[Proof of Proposition~\ref{prop:RateFunction}]
Let $u_x\in\mathcal{D}^o_\Lambda$ be the unique solution of $\Lambda'\left(u\right)=x$,
which exists by Remark~\ref{rem:LambdaStar}~(i).
It is clear from definition~\eqref{eq:DefFenchelLegendreTransf} that, 
for $x\in\Lambda'\left(\mathcal{D}^o_\Lambda\right)$, the Fenchel-Legendre $\Lambda^*$ takes the form given in the proposition.

Assume now that $\Lambda_-'(1)$ is finite.
This is equivalent to $\chi(1)\neq0$ which implies that for every $u\in\mathcal D_\Lambda^o$  we have $u<1$.
Then for any 
$x\in\left[\Lambda'_-\left(1\right),\infty\right)\cap\left(\mathbb{R}\backslash\Lambda'\left(\mathcal{D}^o_\Lambda\right)\right)$
the inequality $\Lambda_-\left(1\right)-\Lambda\left(u\right)\leq x\left(1-u\right)$ 
holds by the Lagrange theorem (and the fact that $\Lambda'$ is strictly increasing).
Hence formula~\eqref{eq:FormLambdaStar} follows. 

If $\Lambda_+'(0)$ is finite, then for every $u\in\mathcal D_\Lambda^o$ we have $u>0$.
For any
$x\in\left(-\infty,\Lambda'_+\left(0\right)\right]\cap\left(\mathbb{R}\backslash\Lambda'\left(\mathcal{D}^o_\Lambda\right)\right)$
the inequality $ux-\Lambda\left(u\right)\leq-\Lambda_+\left(0\right)$ holds for all $u\in\mathcal D_\Lambda^o$.
Hence formula~\eqref{eq:FormLambdaStar} follows. 

The function $\Lambda^*$ is continuously differentiable on $\mathbb{R}$ by~\eqref{eq:FormLambdaStar} and Remark~\ref{rem:LambdaStar}~(i).
Note that, if $0\in\mathcal D_\Lambda^o$, at the minimum we have $u_x=0$.
This implies by definition that the minimum of $\Lambda^*$ is attained at $\Lambda'(0)=x$.
The case $0\notin\mathcal D_\Lambda^o$ follows in a similar way.

If $1\in\mathcal D_\Lambda^o$, then by differentiating the formula in~\eqref{eq:FormLambdaStar} 
we find that the minimum of $x\mapsto \Lambda^*(x)=x$ is attained if and only if $u_x=1$,
which is equivalent to $\Lambda'(1)=x$.
If $1\notin\mathcal D_\Lambda^o$, it is easy to see that the minimum is attained for all $x\geq \Lambda_-'(1)$.
This concludes the proof.
\end{proof}

Before stating the main theorem of this paper, let us define a probability measure $\widetilde{\mathbb{P}}$,
known as the Share measure\label{def:ShareMeasure}, via the Radon-Nikodym derivative
$\D \widetilde{\mathbb{P}}/\D\mathbb{P}$ which at time $t$ takes the form $\E^{X_t}$.
Since $(\E^{X_t})_{t\geq0}$ is a martingale,
$\widetilde{\mathbb{P}}$ is a well-defined probability measure. 
The cumulant generating functions and consequently the Fenchel-Legendre transforms of $X$ under $\mathbb{P}$ and $\widetilde{\mathbb{P}}$
are related by
\begin{equation}
\label{eq:LambdaStarInPTilde}
\widetilde{\Lambda}(u)=\Lambda(u+1),
\quad \text{for all $u$ such that }(1+u)\in\mathcal D_\Lambda,\quad\text{ and }\quad
\widetilde{\Lambda}^*(x)=\Lambda^*(x)-x,\quad \text{for all }x\in\mathbb{R}.
\end{equation}
We are now equipped to state the main theorem of this paper, the proof of which is postponed to Section~\ref{sec:ProofLDPThm}.
\begin{theorem}\label{thm:LDPAffineThm}
The family of random variables $(X_t/t)_{t\geq 1}$ where the process  $X$ is defined in~\eqref{eq:DefAffine}
satisfies a large deviations principle under $\mathbb{P}$ (respectively under $\widetilde{\mathbb{P}}$) with rate function $\Lambda^*$ described in Proposition~\ref{prop:RateFunction}
(resp. $\widetilde{\Lambda}^*$ in~\eqref{eq:LambdaStarInPTilde}) .
\end{theorem}

\section{Asymptotics of option prices and implied volatilities}
\label{sec:RateFunctionsAndOptionPrices}
In this section we relate the rate function $\Lambda^*$ governing the large deviations of the family $(X_t/t)_{t\geq1}$ 
to the option prices in the case of model~\eqref{eq:DefAffine} and the Black-Scholes model.
These asymptotic option prices will then be translated into implied volatility asymptotics.

\subsection{Asymptotics of option prices}
Theorem~\ref{thm:OptionPrices} and Corollary~\ref{cor:BS} below describe the limiting behaviour of European option prices 
respectively in the model~\eqref{eq:DefAffine} and in the Black-Scholes model
when the maturity tends to infinity.
These results were proved in~\cite{JKRM} and we recall them here 
to highlight the importance of proving a large deviations principle 
under both probability measures $\mathbb{P}$ and $\widetilde{\mathbb{P}}$.
\begin{theorem}
\label{thm:OptionPrices}
Let the Fenchel-Legendre transform $\Lambda^*$ be as in~\eqref{eq:DefFenchelLegendreTransf} 
for the family of random variables $(X_t/t)_{t\geq1}$, where $(X_t)_{t\geq0}$
is given by SDE~\eqref{eq:DefAffine}, 
and let $x\in\mathbb{R}$ be a fixed number.
\begin{enumerate}
\item[(i)] 
If $(X_t/t)_{t\geq1}$ satisfies the LDP under the measure $\mathbb{P}$ with the good rate function $\Lambda^*$,
the asymptotic behaviour of a put option with strike $\E^{xt}$ is given by the following formula
$$
\lim_{t\to\infty}\frac{1}{t}\log\mathbb{E}\left[\left(\E^{xt}-\E^{X_t}\right)^+\right] = 
\left\{
\begin{array}{ll}
x-\Lambda^*\left(x\right),\quad & \text{if }x\leq\Lambda_+'\left(0\right),\\
x-\Lambda_+\left(0\right),\quad & \text{if }x>\Lambda_+'\left(0\right),
\end{array}
\right.
$$
where $\Lambda_+\left(0\right)$ and $\Lambda'_+(0)$ are defined in~\eqref{eq:DefLimits}.
\item[(ii)] 
If $(X_t/t)_{t\geq1}$ satisfies the LDP under the measure $\widetilde{\mathbb{P}}$ with the good rate function $\widetilde{\Lambda}^*$,
the asymptotic behaviour of a call option, struck at $\E^{xt}$, is given by
$$
\lim_{t\to\infty}\frac{1}{t}\log\mathbb{E}\left[\left(\E^{X_t}-\E^{xt}\right)^+\right] = 
\left\{
\begin{array}{ll}
x-\Lambda^*\left(x\right),\quad & \text{if }x\geq\Lambda_-'\left(1\right),\\
-\Lambda_-\left(1\right),\quad & \text{if }x<\Lambda_-'\left(1\right), 
\end{array}
\right.
$$
\item[(iii)] 
If $(X_t/t)_{t\geq1}$ satisfies the LDP under both $\mathbb{P}$ and $\widetilde{\mathbb{P}}$ with the respective good rate functions
$\Lambda^*$ and $\widetilde{\Lambda}^*$, the asymptotic behaviour of a covered call option with payoff
$\E^{X_t}-\left(\E^{X_t}-\E^{xt}\right)^+$ is given by 
$$
\lim_{t\to\infty}t^{-1}\log\left(1-\mathbb{E}\left[\left(\E^{X_t}-\E^{xt}\right)^+\right]\right) = 
x-\Lambda^*\left(x\right),\quad  \text{if }x\in\left[\Lambda'_+\left(0\right),\Lambda_-'\left(1\right)\right].
$$
\end{enumerate}
\end{theorem}

Let us consider the Black-Scholes model where the process $\left(X_t\right)_{t\geq 0}$ satisfies the SDE
$\D X_t=-\Sigma^2/2\D t+\Sigma\D W_t$, with $\Sigma>0$. 
Its limiting cumulant generating function reads
$\Lambda_{\mathrm{BS}}(u)=u\left(u-1\right)\Sigma^2/2$
for all $u\in\mathbb{R}$, 
and we define its Fenchel-Legendre transform~\eqref{eq:DefFenchelLegendreTransf} $\Lambda^*_{\mathrm{BS}}(\cdot,\Sigma)$.
Since the function $\partial_x \Lambda'_{\mathrm{BS}}(\cdot,\Sigma)$ is strictly increasing on the whole real line, 
the equation $\Lambda'_{\mathrm{BS}}\left(u\right)=x$ has a unique solution $u_x\in\mathbb{R}$ 
for any real number $x$.
It is straightforward to see that 
$u_x=x/\Sigma^2+1/2$ and hence
$\Lambda^*_{\mathrm{BS}}\left(x,\Sigma\right)=\left(x+\Sigma^2/2\right)^2/\left(2\Sigma^2\right)$ 
for all $x\in\mathbb{R}$.
From this characterisation it is immediate to see that 
$\partial_x \Lambda^*_{\mathrm{BS}}\left(x,\Sigma\right)=0$ if and only if $x=-\Sigma^2/2$
and
$\partial_x \Lambda^*_{\mathrm{BS}}\left(x,\Sigma\right)=1$ if and only if $x=\Sigma^2/2$.
\begin{corollary}\label{cor:BS}
Under the Black-Scholes model, we have the following option price asymptotics.
\begin{align*}
\lim_{t\to\infty}\frac{1}{t}\log\mathbb{E}\left(\E^{xt}-\E^{X_t}\right)_+ & = 
\left\{
\begin{array}{ll}
x-\Lambda^*_{\mathrm{BS}}\left(x,\Sigma\right),\quad & \text{if }x\leq-\Sigma^2/2,\\
x,\quad & \text{if }x>-\Sigma^2/2,
\end{array}
\right.\\
\lim_{t\to\infty}\frac{1}{t}\log\mathbb{E}\left(\E^{X_t}-\E^{xt}\right)_+ & = 
\left\{
\begin{array}{ll}
x-\Lambda^*_{\mathrm{BS}}\left(x,\Sigma\right),\quad & \text{if }x \geq \Sigma^2/2,\\
0,\quad & \text{if }x<\Sigma^2/2,
\end{array}
\right.\\
\lim_{t\to\infty}\frac{1}{t}\log\left(1-\mathbb{E}\left(\E^{X_t}-\E^{xt}\right)_+\right) & = 
\left\{
\begin{array}{ll}
2x+\Sigma^2,\quad & \text{if }x\leq-3 \Sigma^2/2,\\
x-\Lambda^*_{\mathrm{BS}}\left(x,\Sigma\right),\quad & \text{if }x\in\left(-3 \Sigma^2/2,\Sigma^2/2\right],\\
0,\quad & \text{if }x>\Sigma^2/2.
\end{array}
\right.
\end{align*}
\end{corollary}

\subsection{Implied volatility asymptotics}
\label{subsec:ImpliedVol}
We now translate the large-maturity asymptotics for option prices proved above to the study of the implied volatility. 
Proposition~\ref{prop:ImpliedVol} provides the limit of the implied volatility 
for continuous affine stochastic volatility models~\eqref{eq:DefAffine}.
For any real number $x$, let $\sigma_t(x)$ represent the Black-Scholes implied volatility of a European call option with strike price 
$S_0\E^{xt}$ in the model~\eqref{eq:DefAffine}.
Let us further define the function $\sigma_\infty:\mathbb{R}\to\mathbb{R}_+$ by
\begin{equation}\label{eq:sigmainf}
\sigma^2_{\infty}(x):=
2 \left(2\Lambda^*\left(x\right)-x+
\mathcal{I}\left(x\right)\Big(\Lambda^*\left(x\right)\left(\Lambda^*\left(x\right)-x\right)\Big)^{1/2}\right),
\quad\text{for all } x\in\mathbb{R},
\end{equation}
where the function $\mathcal{I}:\mathbb{R}\to\mathbb{R}$ is given by 
$$\mathcal{I}\left(x\right)
:=2\left(\ind_{\left\{x\in\left(\Lambda'_+\left(0\right),\Lambda'_-\left(1\right)\right)\right\}}
+\sgn\left(\chi(0)\right)\ind_{\left\{x<\Lambda'_+(0)\right\}}
+\sgn\left(\chi(1)\right)\ind_{\left\{x>\Lambda'_-\left(1\right)\right\}}\right),$$
with $\sgn(x)=1$ if $x\geq 0$ and $-1$ otherwise,
and where the function $\Lambda^*$ is defined in~\eqref{eq:FormLambdaStar}.
The following proposition gives the behaviour of the implied volatility $\sigma_t$ 
as $t$ tends to infinity for all affine stochastic volatility models with continuous paths.
In~\cite{FJ09} and~\cite{JKRM}, the quantities $\chi(0)$ and $\chi(1)$ are assumed to be strictly negative, 
and hence the function $\mathcal{I}$ here is more general than the function $\mathcal{I}$ in these two papers.
\begin{proposition}\label{prop:ImpliedVol}
The function $\sigma_\infty$ defined in~\eqref{eq:sigmainf} is continuous and the equality
$\lim\limits_{t\to\infty}\sigma_t\left(x\right)=\sigma_\infty\left(x\right)$ 
holds for all $x\in\mathbb{R}$ if the parameter $b$ in the model~\eqref{eq:DefAffine} is not null.
\end{proposition}
\begin{proof}
From Theorem~\ref{thm:OptionPrices} and Corollary~\ref{cor:BS}, the implied volatility $\sigma_\infty$
satisfies the quadratic equation
\begin{equation}\label{eq:BSvolEquation}
\Lambda^*(x)=\Lambda^*_{\mathrm{BS}}\left(x,\sigma_\infty(x)\right),
\end{equation}
for all real number $x$.
The proof of the corollary therefore consists of (a) finding the correct root of this quadratic equation and 
(b) proving the the function $\sigma_t(x)$ converges to this root for all $x$ in the corresponding subset of the real line.
The proof is analogous to the proof of~\cite[Theorem 14]{JKRM}, and we therefore omit it for brevity.
We also refer the reader to the recent work~\cite{GaoLee} for the general methodology to transform option price asymptotics into implied volatility asymptotics.
\end{proof}

\begin{remark}\label{rem:DegenerateCase}
From Corollary~\ref{cor:LimitLambda}, the case $b=0$ can be handled directly since 
the limiting cumulant generating function reads $\Lambda(u)=\frac{1}{2}au\left(u-1\right)$,
for all $u\in\mathcal{D}_{\Lambda}=[0,u_+]$, where $u_+$ is given in~\eqref{eq:DefUP}.
Proposition~\ref{prop:RateFunction} also implies that 
$\Lambda^*(x)=0$ for all $x<\Lambda'_+(0)=-a/2$ and 
$\Lambda^*(x)=\Lambda^*_{\mathrm{BS}}\left(x,\sqrt{a}\right)$ otherwise.
Therefore the limiting implied variance $\sigma^2_{\infty}(x)$ is equal to $-2x$ for all 
$x<-a/2$ and is equal to $a$ for all $x\geq-a/2$.
Note that in the case $\chi(1)=0$, the effective domain $\mathcal{D}_\Lambda$ reads $[u_-,1]$, 
where $u_-$ is given in~\eqref{eq:DefUP}, but the function $\Lambda$ is steep at the right boundary of the domain.
\end{remark}

\subsection{Convergence of the implied volatility of the Heston model to SVI}
In~\cite{GatheralSVI}, Gatheral proposed the so-called `Stochastic Volatility Inspired' (SVI) parameterisation 
of the implied volatility smile.
Using the closed-form representation of the rate function $\Lambda^*$ 
(Proposition~\ref{prop:RateFunction} and Equation~\eqref{eq:Ux}) in the Heston model $a=0$,
Gatheral and Jacquier~\cite{GJ10} proved that this parameterisation
was indeed the true limit of the Heston implied volatility smile as the maturity tends to infinity 
for strikes of the form $S_0\E^{xt}$, 
whenever both conditions $\chi(0)<0$ and $\chi(1)<0$ are met.
Corollary~\ref{cor:SVIHeston} below extends their result without these conditions.
Its proof follows from straightforward manipulations of Formula~\eqref{eq:sigmainf} and we therefore omit it.
Recall that the SVI parameterisation for the implied variance reads
\begin{equation}\label{eq:SVI}
\sigma^2_\mathrm{SVI}\left(x\right)=\frac{\omega_1}{2}\,\left(1+\omega_2\rho x+\sqrt{\left(\omega_2
x+\rho\right)^2+1-\rho^2}\right),
\qquad\text{for all } x\in\mathbb{R},
\end{equation}
where $(\omega_1,\omega_2) \in\mathbb{R}^2$ and $\rho\in\left[-1,1\right]$.
Let us further define the mappings
\begin{equation}\label{eq:ChangeVar}
\omega_1 :=\frac{4b}{\alpha\left(1-\rho^2\right)}
\left(\sqrt{\left(2\beta+\rho\sqrt{\alpha}\right)^2+\alpha\left(1-\rho^2\right)}+\left(2\beta+\rho\sqrt{\alpha}\right)\right)
\qquad\text{and}\qquad
\omega_2 :=\frac{\sqrt{\alpha}}{b}.
\end{equation}

\begin{corollary}\label{cor:SVIHeston}
If $a=0$ and $b\ne 0$, the asymptotic implied volatility $\sigma_\infty$ in~\eqref{eq:sigmainf} simplifies as follows:
\begin{enumerate}[(i)]
\item for all $x\in\Lambda'\left(\mathcal{D}^o_\Lambda\right)$, under the mappings~\eqref{eq:ChangeVar}, 
$\sigma_\infty^2\left(x\right)=\sigma^2_{\mathrm{SVI}}\left(x\right)$;
\item if $\chi(1)>0$, define $\lambda_1:=\sqrt{2b\chi(1)}$, then
$$\sigma^2_\infty(x)=2x+\frac{4\lambda_1}{\alpha}\left(\lambda_1+\sqrt{\lambda_1^2+\alpha x}\right),
\qquad\text{for all }x>\Lambda'_-(1);$$
\item if $\chi(0)>0$, define $\lambda_0:=\sqrt{2b\chi(0)}$, then
$$\sigma^2_\infty(x)=-2x+\frac{4\lambda_0}{\alpha}\left(\lambda_0+\sqrt{-\left(\lambda_0^2+\alpha x\right)}\right),
\qquad\text{for all }x<\Lambda'_+(0);$$
\end{enumerate}
\end{corollary}

\begin{remark}\text{}
\begin{enumerate}[(a)]
\item The case $b=0$ was treated in Remark~\ref{rem:DegenerateCase}.
\item The interval $\Lambda'\left(\mathcal{D}^o_\Lambda\right)$ corresponds to the subset of the real line where the function
$\Lambda^*$ is strictly convex.
When $\chi(1)<0$ and $\chi(0)<0$ (as in~\cite{GJ10}), this interval corresponds to the whole real line.
\item When $a=0$, the quantities in Remark~\ref{rem:Values01} simplify to
\begin{align*}
\Lambda_+(0) & = -\frac{2b\beta}{\alpha},\quad
& \Lambda_+'(0) & = -\frac{b}{2\sqrt{\alpha}}\left(4\rho+\frac{\sqrt{\alpha}}{\beta}\right),
\quad & \text{when }\chi(0)>0,\\
\Lambda_-(1) & = -\frac{2b}{\alpha}\left(\beta+\rho\sqrt{\alpha}\right),\quad
& \Lambda_-'(1) & = -\frac{b}{2\sqrt{\alpha}}\left(4\rho+\frac{\sqrt{\alpha}}{\beta+\rho\sqrt{\alpha}}\right),
\quad & \text{when }\chi(1)>0.
\end{align*}
\end{enumerate}
\end{remark}

\begin{remark}
Note that $\omega_1$ in~\eqref{eq:ChangeVar} is a continuous function of $\rho\in(-1,1)$ and has the following limits:
$$
\begin{array}{rlll}
\displaystyle 
\omega_1 & = -2b/\left(2\beta+\sqrt{\alpha}\right),
\quad & \text{if }2\beta+\rho\sqrt{\alpha}<0,
\quad & \text{when } \rho=1,\\
\omega_1 & = -2b/\left(2\beta-\sqrt{\alpha}\right),
\quad & \text{if }2\beta+\rho\sqrt{\alpha}<0,
\quad & \text{when } \rho=-1,
\end{array}
$$
It diverges to $\pm\infty$ in the other cases. 
In terms of the SVI implied volatility smile, whenever $\rho\in\{-1,1\}$, 
we can plug these limits when they exist into~\eqref{eq:SVI}, or simplify directly~\eqref{eq:sigmainf} using~\eqref{eq:upmrho1},
and we obtain
$$
\sigma^2_{\mathrm{SVI}}(x) = 
\left\{
\begin{array}{ll}
\displaystyle -2\frac{b+\rho x\sqrt{\alpha}}{2\beta+\rho\sqrt{\alpha}},
\quad & \text{if }b+\rho x\sqrt{\alpha}>0,\\
0,
\quad & \text{if }b+\rho x\sqrt{\alpha}\leq 0.
\end{array}
\right.
$$
When $\chi(0)<0$ and $\chi(1)<0$, this is consistent 
with the fact---see~\cite[Proposition 5]{Zeliade}---that for any maturity the implied volatility is decreasing 
(resp. increasing) whenever the correlation parameter $\rho$ is equal to $-1$ (resp. equal to $1$).
In the case $\rho=-1$, the proof of this statement in~\cite[Proposition 5]{Zeliade} is based on the following remark:
if $(X_t)_{t\geq 0}$ satisfies the SDE~\eqref{eq:DefAffine}, then It\^o's formula gives
$$X_t=X_0-\frac{1}{2}\int_{0}^{t}V_s \D s+\frac{bt+V_0}{\sqrt{\alpha}}-\frac{1}{\sqrt{\alpha}}\left(V_t-\beta\int_{0}^{t}V_s \D s\right),
\quad\text{for any }t\geq 0.$$
When $\beta\leq 0$, since the variance process $\left(V_t\right)_{t\geq 0}$ is not negative, it is clear that for any $t\geq 0$, the random variable $X_t$ is bounded above, and hence, the implied volatility is null above this level.
As soon as $\beta$ is strictly positive, this bound does not hold anymore and the implied volatility is not flat any more.
Note further than the condition $\chi(1)\geq 0$ implies the inequality $\chi(0)\geq 0$ when $\rho=-1$.
In the Heston model, this implies that only Case (i) in Corollary~\ref{cor:SVIHeston} applies, 
i.e. the SVI parameterisation holds on the whole real line.
The case $\rho=1$ is symmetric (under the Share measure) and we omit an analogous discussion.
\end{remark}

\section{Proof of Theorem~\ref{thm:LDPAffineThm}}\label{sec:ProofLDPThm}
We split the proof of the theorem according to the four cases arising in Corollary~\ref{cor:LimitLambda}.
In the case~(i)~(a), since the limiting cumulant generating function $\Lambda$ is differentiable and essentially smooth 
in the interior of its domain $\mathcal{D}_{\Lambda}$ and $0\in\mathcal{D}_{\Lambda}$ (Corollary~\ref{cor:LimitLambda}), 
then the theorem follows by a direct application of the G\"artner-Ellis theorem. 
This case was already proved when $a=0$ in~\cite{FJ09} and when $a\ne 0$---albeit in a more general framework---in~\cite{JKRM}.
In the case (i) (b), the effective domain $\mathcal{D}_\Lambda$ is 
$\left[u_-,1\right]$ with $u_-<0$, but the function $\Lambda$ is not steep at the right boundary, and hence the G\"artner-Ellis theorem
does not apply.  
Proposition~\ref{prop:NonSteep1} shows that a full LDP however still holds in this case.
The proof of this theorem relies on Lemma~\ref{lem:LambdaExpansion} and Lemma~\ref{lem:WeakConvergenceAt1}.
Lemma~\ref{lem:LambdaExpansion} concerns the behaviour of the function 
$\Lambda_t$ in~\eqref{eq:LambdaT} around $1$ as $t$ tends to infinity
and Lemma~\ref{lem:WeakConvergenceAt1} is a weak convergence result for the process $\left(X_t\right)_{t>0}$.
For sake of clarity, we postpone these lemmas and their proofs to Appendix~\ref{app:TechLemma}.
Proposition~\ref{prop:NonSteep0} deals with the case where $\mathcal{D}_\Lambda=[0,u_+]$ with $u_+>1$ 
and Proposition~\ref{prop:NonSteep0} states a LDP when $\mathcal{D}_\Lambda=[0,1]$.
By a shifting  argument, Theorem~\ref{thm:LDPAffineThm} clearly holds under $\widetilde{\mathbb{P}}$
as soon as a large deviations principle is satisfied in all cases under~$\mathbb{P}$.
We therefore state the three propositions below under the measure $\mathbb{P}$.

\begin{proposition}\label{prop:NonSteep1}
In case (i)(b), the family $\left(X_t/t\right)_{t>0}$ satisfies a LDP under $\mathbb{P}$ with rate function $\Lambda^*$.
\end{proposition}
\begin{proposition}\label{prop:NonSteep0}
In case (ii)(a), the family $\left(X_t/t\right)_{t>0}$ satisfies a LDP under $\mathbb{P}$ with rate function $\Lambda^*$.
\end{proposition}
\begin{proposition}\label{prop:LDPiib}
In case (ii)(b), the family $\left(X_t/t\right)_{t>0}$ satisfies a LDP under $\mathbb{P}$ with rate function $\Lambda^*$.
\end{proposition}

\begin{remark}\text{}
\begin{enumerate}[(i)]
\item In the case $\chi(1)=0$, the domain $\mathcal{D}_\Lambda$ of the limiting cumulant generating function $\Lambda$ is $[u_-,1]$ 
and the function is steep at the right boundary $u_+=1$ and therefore the G\"artner-Ellis theorem holds.
However, under the Share measure defined on Page~\pageref{def:ShareMeasure}, the origin is in $\mathcal{D}_\Lambda$  but not in its interior.
\item In view of Remark~\ref{rem:SimpleFacts} (III) and Corollary~\ref{cor:LimitLambda}, 
the origin is not in the interior of $\mathcal{D}_\Lambda$ when $\chi\left(0\right)\geq 0$.
\end{enumerate}
\end{remark}

\begin{notation}
For any $t>0$, we shall denote by $\mathbb{P}_t$ the law of the random variable $X_t$.
\end{notation}

\begin{proof}[Proof of Proposition~\ref{prop:NonSteep1}]
In the case $\mathcal{D}_\Lambda=\left[u_-,1\right]$ with $u_-<0$, 
the limiting cumulant generating function $\Lambda$ is not steep at the right boundary $1$ any more.
The upper bound holds for compact sets in $\mathbb{R}$ by Chebychev inequality, 
and its extension to closed sets is a consequence of the origin being inside the interior of the domain 
of the limiting log Laplace transform $\Lambda$.
These arguments are the same as in the proof of the G\"artner-Ellis theorem~\cite[Section 2.3]{DemboZeitouni}.

We now prove the lower bound for the $\liminf$ on open sets in $\mathbb{R}$.
The set of exposed points of the function $\Lambda$ is the interval 
$\left(-\infty,\Lambda_-'\left(1\right)\right)$
so that the lower bound for open sets in this interval follows from the G\"artner-Ellis theorem.
We therefore consider $x\geq\Lambda'_-\left(1\right)$ from now on.
Since the function $\Lambda$ is continuously differentiable and convex on $\mathcal{D}^o_\Lambda$, 
two possible cases arise: either it attains its minimum at a unique point $u_0\in\mathcal{D}^o_\Lambda$,
and hence $\Lambda_{-}'(1)>0$, or
it is strictly decreasing on its effective domain, which implies $\Lambda_{-}'(1)\leq 0$.
In the case $\Lambda_{-}'(1)>0$, we can define a new probability measure $\mathbb{P}_t^0$ for each $t>0$ via
$$\frac{\D\mathbb{P}_t^0}{\D\mathbb{P}_t}(z)
:=\exp\Big(u_0 z t-\Lambda_t(u_0)\Big),
\qquad\text{for any }z\in\mathbb{R}.$$
The proof of the lower bound then follows exactly as in the standard 
G\"artner-Ellis theorem with this change of measure.
It can similarly be shown that since $\Lambda$ is strictly convex on $\mathcal{D}^o_{\Lambda}$,
the measure $\mathbb{P}_t^0$ converges weakly to a Gaussian random measure
with zero mean and variance $\Lambda''(u_0)$.

We now consider the case $\Lambda_{-}'(1)\leq 0$. 
As in the G\"artner-Ellis theorem, it suffices to prove the equality
$$\lim_{\delta\to 0}\liminf_{t\to\infty}t^{-1}\log\mathbb{P}\left(X_t/t\in\left(x,x+\delta\right)\right)
\geq -\Lambda^*\left(x\right).$$
In view of Lemma~\ref{lem:LambdaExpansion}, let us define the function 
$\overline{\Lambda}_t:\mathcal{D}_t\cap\left(-\infty,1\right)$ by
\begin{equation}\label{eq:LambdaBar}
\overline{\Lambda}_t(u):=\Lambda\left(u\right)t-\frac{2b}{\alpha}\log\left(1-u\right).
\end{equation}
The key ingredient now is to remark that, for each $t>0$, the function 
$u\mapsto t^{-1}\overline{\Lambda}_t(u)$ is smooth and convex in the interval $(0,1)$ and furthermore is steep at $1$.
Therefore for any $t>0$, there exists a unique solution $u_t$ to the equation
$\overline{\Lambda}'_t\left(u_t\right)=0$.
Using similar arguments as in~\cite{Florens}, it is clear that $u_t$ converges to $1$ from below as $t$ tends to infinity.
Let us further define a new measure $\overline{\mathbb{P}}_t$ by
\begin{equation}\label{eq:ProbaPBar}
\frac{\D\overline{\mathbb{P}}_t}{\D\mathbb{P}_t}\left(z\right)
:=\exp\Big(u_t z t-\Lambda_t\left(u_t\right)\Big),
\quad\text{for any }z\in\mathbb{R}.
\end{equation}
For any $\delta>0$ we then have
\begin{align}\label{eq:NonSteep1}
t^{-1}\log\mathbb{P}_t\Big(X_t/t\in\left(x,x+\delta\right)\Big)
 & = t^{-1}\log\int_{\left(x,x+\delta\right)}\exp\Big(\Lambda_t\left(u_t\right)-u_t zt\Big)\D\overline{\mathbb{P}}_t\left(z\right)\nonumber\\
 & = t^{-1}\Lambda_t\left(u_t\right)-u_t x+t^{-1}\log\int_{\left(x,x+\delta\right)}\E^{-u_t\left(z-x\right)t}\D\overline{\mathbb{P}}_t\left(z\right)\\
 & \geq t^{-1}\Lambda_t\left(u_t\right)-u_t\left(x+\delta\right)+t^{-1}\log\overline{\mathbb{P}}_t\Big(X_t/t\in\left(x,x+\delta\right)\Big),
\nonumber
\end{align}
for $t$ large enough so that $u_t>0$, and hence
\begin{align}\label{eq:LowerBoundNonSteep}
\lim_{\delta\to 0}\liminf_{t\to\infty}t^{-1}\log\mathbb{P}\Big(X_t/t\in\left(x,x+\delta\right)\Big)
  & \geq \liminf_{t\to\infty}\left(t^{-1}\Lambda_t\left(u_t\right)-u_t x\right)\\
  & + \lim_{\delta\to 0}\liminf_{t\to\infty}t^{-1}\log\overline{\mathbb{P}}_t
\Big(X_t/t\in\left(x,x+\delta\right)\Big)\nonumber
\end{align}
We now have to find a lower bound for both terms on the right-hand side of this inequality.
Since the function $\Lambda_t$ is convex for all $t>0$, 
we have 
$\Lambda_t\left(u_t\right)-\Lambda_t\left(u\right)
\geq\left(u_t-u\right)\Lambda'_t\left(u\right)$ for all $u<1$. 
From~\cite[Theorem 25.7]{Rockafellar} we have 
$\lim_{t\to\infty}t^{-1}\Lambda'_t\left(u\right)=\Lambda'\left(u\right)$ for all $u<1$ and
$\liminf_{t\to\infty}t^{-1}\Lambda_t\left(u_t\right)\geq\Lambda\left(u\right)+\left(1-u\right)\Lambda'\left(u\right)$,
which implies that $\liminf_{t\to\infty}t^{-1}\Lambda_t\left(u_t\right)\geq\Lambda_-\left(1\right)$. 
The fact that $u_t$ converges to $1$ as $t$ tends to infinity and the characterisation of 
the Fenchel-Legendre transform $\Lambda^*$ in Proposition~\ref{prop:RateFunction} gives
$$\liminf_{t\to\infty}\left(t^{-1}\Lambda\left(u_t\right)-u_t x\right)\geq \Lambda^*\left(x\right).$$
When $\Lambda'_-(1)<0$, Lemma~\ref{lem:WeakConvergenceAt1} implies that $\overline{\mathbb{P}}_t$ converges to a probability measure $\overline{\mathbb{P}}$ with full support
as $t$ tends to infinity, and therefore the last term on the right-hand side of the 
inequality~\eqref{eq:LowerBoundNonSteep} tends to zero as $t$ tends to infinity (for any $\delta>0$).
This proves the theorem in the case $\Lambda'_-(1)<0$.\\
When $\Lambda'_-(1)=0$, we cannot conclude immediately since Lemma~\ref{lem:WeakConvergenceAt1}
is a convergence result for the family $\left(X_t/\sqrt{t}\right)_{t>0}$ and we need a convergence property for
the family $\left(X_t/t\right)_{t>0}$.
However, we can argue as follows.
Let $(\xi_t)_{t\geq 0}$ be an independent L\'evy process with L\'evy exponent $\phi$ defined on a domain $\mathcal{D}_\phi$
strictly containing $\mathcal{D}_\Lambda$ and such that $\phi'(1)\ne 0$.
Consider now the random variable $Y_t:=X_t+\xi$.
The moment generating function of $Y$ is then
$$\Lambda_t^Y(u):=\log\mathbb{E}\left(\E^{uY_t}\right)
=\Lambda_t(u)+\phi(u)t,$$
for any $t\geq 0$ and any $u\in\mathcal{D}_t$.
Therefore 
$$\Lambda^Y(u):=\lim_{t\to\infty}t^{-1}\Lambda_t^Y(u)=\Lambda(u)+\phi(u),
\qquad\text{for all }u\in\mathcal{D}_{\Lambda},$$
where $\mathcal{D}_\Lambda$ is characterised in Corollary~\ref{cor:LimitLambda}.
In particular, note that 
$$\partial_u\Lambda^Y_- (1):=\lim_{u\nearrow 1}\partial_u\Lambda^Y(u) = \Lambda'_-(1)+\phi'(1).$$
Note that $\Lambda'_-(1)=0$ implies that $\partial_u\Lambda^Y_- (1)\ne 0$.
Since the effective domain of the limiting cumulant generating function of $Y$ is the same as that of $X$, 
we therefore obtain a large deviations principle for the family $\left(t^{-1}Y_t\right)_{t>0}$ as $t$ tends to infinity using the analysis above.
If the two families $\left(t^{-1}X_t\right)_{t>0}$ and $\left(t^{-1}Y_t\right)_{t>0}$
are exponentially equivalent, then the LDP for $\left(t^{-1}Y_t\right)_{t>0}$ implies the LDP for $\left(t^{-1}X_t\right)_{t>0}$ by~\cite[Theorem 4.2.13]{DemboZeitouni}.
Recall that two families are said to be exponentially equivalent if for all $\delta>0$, 
$$\limsup_{t\to \infty}t^{-1}\log\mathbb{P}\left(\frac{\left|X_t-Y_t\right|}{t}>\delta\right)=-\infty.$$
Since 
$\mathbb{P}\left(\frac{\left|X_t-Y_t\right|}{t}>\delta\right) = 
\mathbb{P}\left(|\xi_t|>\delta t\right)$, we simply need to find a (L\'evy process) 
satisfying $\mathbb{P}\left(|\xi_t|>\delta t\right)\sim t^{-\beta}\exp\left(-\alpha t^\gamma\right)$,
for some $\alpha>0$, $\beta>0$ and $\gamma>1$ as $t$ tends to infinity.
The existence of such a L\'evy process is given in~\cite[Theorem 26.1, case (i)]{Sato}.
\end{proof}

\begin{remark}
A similar issue arose in~\cite{Bercu} where the authors studied large deviations properties 
for the maximum likelihood estimator of an Ornstein-Uhlenbeck process.
When the limiting cumulant generating function is flat 
at the boundary of the domain, i.e. $\Lambda'_-(1)=0$, 
they showed that the same large deviations principle holds.
Only  the higher-order terms in the asymptotic expansion of the probability change.
\end{remark}

\begin{proof}[Proof of Proposition~\ref{prop:NonSteep0}]
Let us first consider open and closed sets in the set of exposed points $\left(\Lambda'_+(0),\infty\right)$.
By~\cite[Theorem 4.5.3]{DemboZeitouni} we know that the upper bound of the G\"artner-Ellis theorem
holds on compact sets even when the origin is not in the interior of the domain of the limiting cumulant generating function.
In the proof of the G\"artner-Ellis theorem the assumption ensuring that the origin
lies within the interior of $\mathcal{D}_{\Lambda}$ is required 
(i) to derive the upper bound for closed sets and not only for compact sets, and
(ii) to prove that the Fenchel-Legendre transform $\Lambda^*$ of $\Lambda$
is a good rate function.
We know that the function $\Lambda^*$ is not a good convex rate function,
and we shall see how to deal with this. 
Let us first prove (i).
Let $B$ be a Borel set in $\mathbb{R}$. 
We want to prove that
\begin{equation}\label{eq:Inequalities}
-\inf_{z\in B^o}\Lambda^*\left(z\right)
\leq \liminf_{t\to\infty}t^{-1}\log\mathbb{P}\left(X_t/t\in B\right)
\leq \limsup_{t\to\infty}t^{-1}\log\mathbb{P}\left(X_t/t\in B\right)
\leq -\inf_{z\in\overline{B}}\Lambda^*\left(z\right).
\end{equation}
The upper bound for compact subsets of the real line follows from Chebychev inequality, 
and~\cite[Proposition 5.2]{OBrien} shows that this extends to closed sets even when $0\notin\mathcal{D}_\Lambda$.
In the case where we are only interested in intervals 
(i.e. $\mathbb{P}\left(X_t/t\leq x\right)$ or $\mathbb{P}\left(X_t/t\in[y, x]\right)$), 
the following argument is self-contained and does not rely on~\cite{OBrien}:
let $x$ be a real number.
For any $y<x$,  Chebychev inequality implies the following upper bound on the compact interval $\left[y,x\right]$:
\begin{equation}\label{eq:UpperCompact}
\limsup_{t\to\infty}t^{-1}\log\mathbb{P}\left(X_t/t\in\left[y,x\right]\right)\leq -\inf_{z\in\left[y,x\right]}\Lambda^*\left(z\right).
\end{equation}
Since $\chi\left(0\right)>0$, the function $\Lambda^*$ is constant on~$\left(-\infty,\Lambda'_+\left(0\right)\right)$ and strictly increasing outside.
Since we are interested in the limit as $y$ tends to $-\infty$, 
we can consider $y\leq\Lambda^*\left(\Lambda'_+\left(0\right)\right)$ without loss of generality,
and hence 
$\inf_{z\in\left[y,x\right]}\Lambda^*\left(z\right)=\Lambda^*\left(\Lambda'_+\left(0\right)\right)$
always holds for such $y$.
Using the fact that
$\mathbb{P}\left(X_t/t\leq x\right)=\lim_{y\to-\infty}\mathbb{P}\left(X_t/t\in\left[y,x\right]\right)$, 
Inequality~\eqref{eq:UpperCompact} implies that for any $\varepsilon>0$ there exists $t^*\left(\varepsilon\right)>0$ such that
$$\mathbb{P}\left(X_t/t\in\left[y,x\right]\right)\leq
\exp\left(-\Lambda^*\left(\Lambda'_+\left(0\right)\right)t
+\varepsilon t\right),
\quad\text{for any } t>t^*\left(\varepsilon\right).$$
Since the right-hand side does not depend on~$y$ 
we can now take the limit on both sides as $y$ tends to~$-\infty$, and hence
$t^{-1}\log\mathbb{P}\left(X_t/t\leq x\right)\leq
-\Lambda^*\left(\Lambda'_+\left(0\right)\right)+\varepsilon$ holds.
Since $\varepsilon$ can be taken arbitrarily small we obtain 
$\limsup_{t\to\infty}t^{-1}\log\mathbb{P}\left(X_t/t\leq x\right)\leq-\Lambda^*\left(\Lambda'_+\left(0\right)\right)$.
A similar argument leads the upper bound
$\limsup_{t\to\infty}t^{-1}\log\mathbb{P}\left(X_t/t\geq x\right)\leq
-\Lambda^*\left(x\right)\ind_{\left\{x\geq\Lambda'_+\left(0\right)\right\}}
-\Lambda^*\left(\Lambda'_+\left(0\right)\right)\ind_{\left\{x<\Lambda'_+\left(0\right)\right\}}$.

We now want to prove lower bound estimates for the $\liminf$ on open sets of the real line.
Let us consider $x>\Lambda'_+\left(0\right)$ and let $\eta$ be the unique solution to $\Lambda'\left(\eta\right)=x$.
Let us define a new measure $\mathbb{Q}_t$ by
\begin{equation}\label{eq:ProbaQ}
\frac{\D\mathbb{Q}_t}{\D\mathbb{P}}\left(z\right):=\exp\left(\eta z t-\Lambda_t\left(\eta\right)\right),
\quad\text{for any }z\in\mathbb{R}.
\end{equation}
For any $\delta>0$ small enough, denote $B_{x,\delta}$ the open ball centered on $x$ with radius $\delta$, then we have
\begin{align*}
t^{-1}\log\mathbb{P}\Big(X_t/t\in B_{x,\delta}\Big)
 & = t^{-1}\log\int_{B_{x,\delta}}\exp\Big(\Lambda_t\left(\eta\right)-\eta zt\Big)\D\mathbb{Q}_t\left(z\right)\\
 & = t^{-1}\Lambda_t\left(\eta\right)-\eta x+t^{-1}\log\int_{B_{x,\delta}}\exp\left(-\eta\left(z-x\right)t\right)\D\mathbb{Q}_t\left(z\right)\\
 & \geq t^{-1}\Lambda_t\left(\eta\right)-\eta x-\left|\eta\right|\delta+t^{-1}\log\mathbb{Q}_t\Big(X_t/t\in B_{x,\delta}\Big),
\end{align*}
and hence
\begin{align}\label{eq:LowerBoundNonSteep2}
\lim_{\delta\to 0}\liminf_{t\to\infty}t^{-1}\log\mathbb{P}\Big(X_t/t\in B_{x,\delta}\Big)
  & \geq \liminf_{t\to\infty}\left(t^{-1}\Lambda\left(\eta\right)-\eta x\right)
 + \lim_{\delta\to 0}\liminf_{t\to\infty}t^{-1}\log\mathbb{Q}_t\Big(X_t/t\in B_{x,\delta}\Big)\nonumber\\
  & \geq -\Lambda^*\left(x\right)
 + \lim_{\delta\to 0}\liminf_{t\to\infty}t^{-1}\log\mathbb{Q}_t\Big(X_t/t\in B_{x,\delta}\Big).
\end{align}
We now have to find a lower bound for the last term on the right-hand side of this inequality
as $t$ tends to infinity and $\delta$ to zero.
Define the function $\widehat{\Lambda}\left(\cdot\right):=\Lambda\left(\cdot+\eta\right)-\Lambda\left(\eta\right)$.
$\widehat{\Lambda}$ is the limiting logarithmic moment generating function of $\mathbb{Q}_t$ and 
for each $t>0$, we denote $\widehat{\Lambda}_t$ the logarithmic mgf of $\mathbb{Q}_t$, and we have
$$
t^{-1}\widehat{\Lambda}_t\left(\lambda\right) 
 :=t^{-1}\log\int_{\mathbb{R}}\E^{\lambda zt}\D\mathbb{Q}_t\left(z\right)\\
 = t^{-1}\Lambda_t\left(\lambda+\eta\right)-t^{-1}\Lambda_t\left(\eta\right),
$$
which converges to $\widehat{\Lambda}\left(\lambda\right)$ as $t$ tends to infinity.
We now define the Fenchel-Legendre transform $\widehat{\Lambda}^*$ of $\widehat{\Lambda}$ as
$$
\widehat{\Lambda}^*\left(z\right)
:=\sup_{\lambda\in\mathbb{R}}\left\{\lambda z-\widehat{\Lambda}\left(\lambda\right)\right\}
=\Lambda^*\left(z\right)-\eta z+\Lambda\left(\eta\right).
$$
Let $B^c_{x,\delta}$ denote the complement in $\mathbb{R}$ of the open ball $B_{x,\delta}$.
We can now apply the upper bound estimate derived above for the measure $\mathbb{Q}_t$ on the closed set $B^c_{x,\delta}$:
$$\limsup_{t\to\infty}t^{-1}\log\mathbb{Q}_t\left(X_t/t\in B^c_{x,\delta}\right)
\leq-\inf_{z\in B^c_{x,\delta}}\widehat{\Lambda}^*\left(z\right).$$
Since $\Lambda^*\left(x\right)\geq \eta x-\Lambda\left(\eta\right)$ by definition of the Fenchel-Legendre transform, we have
$$
\inf_{z\in B^c_{x,\delta}}\widehat{\Lambda}^*\left(z\right)
=\inf_{z\in B^c_{x,\delta}}\left\{\Lambda^*\left(z\right)-\left(\eta z-\Lambda\left(\eta\right)\right)\right\}
=\inf_{z\in B^c_{x,\delta}}\left\{\sup_{\lambda\in\mathbb{R}}\Big(\lambda z-\Lambda\left(\lambda\right)\Big)-\left(\eta z-\Lambda\left(\eta\right)\right)\right\}
\geq 0.$$
The expression 
$\sup_{\lambda\in\mathbb{R}}\left(\lambda z-\Lambda\left(\lambda\right)\right)-\left(\eta z-\Lambda\left(\eta\right)\right)$
above is always not negative and is null if and only if $\eta=\lambda$.
The strict monotonicity of the function $\Lambda'$ on the interval $\left(\Lambda'_+\left(0\right),\infty\right)$
implies that $z=x$, which is not possible since $z$ takes values only in the complement of the open ball $B_{x,\delta}$.
Hence 
$\inf_{z\in B^c_{x,\delta}}\widetilde{\Lambda}^*\left(z\right)$ is strictly positive for all $x>\Lambda'_+\left(0\right)$.
This implies that 
$\limsup_{t\to\infty}t^{-1}\log\mathbb{Q}_t\left(X_t/t\in B^c_{x,\delta}\right)<0$ and therefore
$\mathbb{Q}_t\left(X_t/t\in B^c_{x,\delta}\right)$ tends to zero and 
$\mathbb{Q}_t\left(X_t/t\in B_{x,\delta}\right)$ tends to one as $t$ tends to infinity for all $\delta>0$.
In particular this implies
$$
\lim_{\delta\to 0}\liminf_{t\to\infty}t^{-1}\mathbb{Q}_t\left(X_t/t\in B^c_{x,\delta}\right)=0,
$$
and the result follows from~\eqref{eq:LowerBoundNonSteep2}.

Consider now open or closed sets in the interval $\left(-\infty,\Lambda'_+(0)\right)$.
The proof of the theorem follows analogous steps as the proof of Proposition~\ref{prop:NonSteep1} 
on sets in $\left(\Lambda'_-(1),\infty\right)$.
We consider a time-dependent change of measure, use an auxiliary convex function $\overline{\Lambda}_t$, 
steep at $0$ and well-defined on $\left(0,\infty\right)\cap\mathcal{D}_t$, for each $t>0$.
This function clearly exists since the function $\Lambda_t$ itself is steep at the left boundary of its effective domain
$\mathcal{D}_t$ which converges to the origin from below.
Lemma~\ref{lem:WeakConvergenceAt0} proves weak convergence results for the random variable $\left(\pi_t^{(0)} X_t/t\right)$, 
where $\pi_t^{(0)}$ is equal to $1$ if $\Lambda'_+(0)>0$ and is equal to $\sqrt{t}$ if $\Lambda'_+(0)=0$.
Therefore, using analogous arguments as in the proof of Proposition~\ref{prop:NonSteep1}, the proposition follows.
\end{proof}

\begin{proof}[Proof of Proposition~\ref{prop:LDPiib}]
The limiting cumulant generating function $\Lambda$ is not steep at either boundary $0$ or $1$.
A large deviations principle clearly holds on any subsets of $\left(\Lambda'_{+}(0),\Lambda'_-(1)\right)$.
For subsets of $\left(\Lambda'_-(1),\infty\right)$, we appeal to Proposition~\ref{prop:NonSteep1} and 
for subsets of $\left(-\infty,\Lambda'_+(0)\right)$, we appeal to Proposition~\ref{prop:NonSteep0}.
\end{proof}

\section{Technical lemmas}\label{app:TechLemma}
\begin{lemma}
\label{lem:LambdaExpansionAt0}
If $\chi\left(0\right)>0$,
the following holds for the function $\Lambda_t$ as $t$ tends to infinity:
$$
t^{-1}\Lambda_t\left(u\right) = \Lambda(u)-\frac{2b}{\alpha t}\log(u)+t^{-1}R_t^{(0)}(u),
\qquad\text{for any } u\in\mathcal{D}_t\cap\left(0,\infty\right),
$$
where for any $t\geq 0$, the function $R^{(0)}_t:\mathcal{D}_t\cap\left(0,\infty\right)\to\mathbb{R}$ is analytic
and converges on any compact subset of $\mathcal{D}_t\cap\left(1,\infty\right)$.
\end{lemma}
\begin{proof}
From~\eqref{eq:DefGamma}, we clearly have 
$f_t\left(u\right)\sim\cosh\left(\gamma\left(u\right)t/2\right)\left|1-\chi\left(u\right)/\gamma\left(u\right)\right|$
as $t$ tends to infinity, so that
$$\log\left(f_t\left(u\right)\right)=\log\left|1-\frac{\chi\left(u\right)}{\gamma\left(u\right)}\right|+\gamma\left(u\right)t/2-\log\left(2\right)
+o\left(\E^{-\gamma\left(u\right)t}\right),
\quad\text{as }t \text{ tends to infinity}.$$
The condition $\chi(0)>0$ implies that
$1-\frac{\chi\left(u\right)}{\gamma\left(u\right)}=\frac{\alpha u}{2\beta^2}+\mathcal{O}\left(u^2\right)$,
and the last term in the expression~\eqref{eq:LambdaT} for $\Lambda_t^H$ clearly satisfy the properties we need for the function 
$R_t^{(0)}$.
\end{proof}

\begin{lemma}
\label{lem:LambdaExpansion}
The following holds for the function $\Lambda_t$ as $t$ tends to infinity:
$$
t^{-1}\Lambda_t\left(u\right) = 
\Lambda\left(u\right)-\frac{2b}{\alpha t}\log\left(1-u\right)+t^{-1}R^{(1)}_t\left(u\right),
\qquad\text{for any } u\in\mathcal{D}_t\cap\left(-\infty,1\right),
$$
where for any $t\geq 0$, the function $R^{(1)}_t:\mathcal{D}_t\cap\left(-\infty,1\right)\to\mathbb{R}$ 
is analytic and converges on any compact subsets of $\mathcal{D}_t\cap\left(-\infty,1\right)$.
\end{lemma}

\begin{proof}
Let us consider the first case 
$\chi\left(0\right)\leq 0$ and $\chi\left(1\right)>0$.
From~\eqref{eq:DefGamma}, we clearly have that 
$f_t\left(u\right)\sim\cosh\left(\gamma\left(u\right)t/2\right)\left|1-\chi\left(u\right)/\gamma\left(u\right)\right|$
as $t$ tends to infinity, so that
$$\log\left(f_t\left(u\right)\right)=\log\left|1-\frac{\chi\left(u\right)}{\gamma\left(u\right)}\right|+\gamma\left(u\right)t/2-\log\left(2\right)
+o\left(\E^{-\gamma\left(u\right)t}\right),
\quad\text{as }t \text{ tends to infinity}.$$
Note further that the assumption $\chi\left(1\right)>0$ implies the expansion
$1-\frac{\chi\left(u\right)}{\gamma\left(u\right)}=-\frac{\alpha}{2}\left(u-1\right)/\chi\left(1\right)^2+\mathcal{O}\left(\left(u-1\right)^2\right)$.
It is straightforward to see that the last term in the expression~\eqref{eq:LambdaT} for $\Lambda_t^H$ 
satisfy the properties we need for the function 
$R_t^{(1)}$.
\end{proof}

Recall that the logarithmic Laplace transform of a 
Gamma-distributed random variable $Y$
with strictly positive parameters $\zeta_1$ and $\zeta_2$ ($Y\sim\Gamma\left(\zeta_1,\zeta_2\right)$) reads
$\log\mathbb{E}\left(\exp\left(uY\right)\right)=-\zeta_1\log\left(1-\zeta_2 u\right)$,
for all $u<1/\zeta_2$. 
We also denote $\delta\left(\zeta\right)$ the distribution of a Dirac random variable with parameter $\zeta$,
$\mathcal{N}\left(\mu,\nu\right)$ a standard Gaussian with mean $\mu$ and variance $\nu$,
and the symbol $\ast$ stands for the convolution operator.

\begin{lemma}\label{lem:WeakConvergenceAt1}
Under the measure $\overline{\mathbb{P}}_t$ defined in~\eqref{eq:ProbaPBar}, the sequence of random variables 
$\left(\pi_t^{(1)} X_t/t\right)_t>0$ 
converges weakly to the random variable $Y$ where
$$
Y\equalDistrib
\left\{
\begin{array}{lll}
\displaystyle 
\delta\left(\Lambda'_-\left(1\right)\right)
\ast\Gamma\left(\frac{2b}{\alpha},-\frac{2b}{\alpha\Lambda'_-\left(1\right)}\right),
\quad & \text{and}\quad \pi_t^{(1)}=1,
\quad & \text{if }\Lambda'_-\left(1\right)<0;\\
\\
\displaystyle 
\mathcal{N}\left(-\sqrt{\frac{2b\Lambda''_-\left(1\right)}{\alpha}},\Lambda''_-\left(1\right)\right)
\ast\Gamma\left(\frac{2b}{\alpha},\sqrt{\frac{2b}{\alpha\Lambda''_-\left(1\right)}}\right),
\quad & \text{and}\quad \pi_t^{(1)}=\sqrt{t},
\quad & \text{if }\Lambda'_-\left(1\right)=0.\\
\end{array}
\right.
$$
\end{lemma}
\begin{proof}
Consider first that $\Lambda'_-\left(1\right)<0$.
For all $t>0$ and all $\xi$ such that $u_t+\xi/t\in\mathcal{D}_t$, we can write
$$
\log\mathbb{E}_{\overline{\mathbb{P}}_t}\left(\exp\left(\xi\frac{X_t}{t}\right)\right)
=\log\mathbb{E}_{\overline{\mathbb{P}}_t}\left(\exp\left(\left(u_t+\frac{\xi}{t}\right)X_t-\Lambda_t\left(u_t\right)\right)\right)
=\Lambda_t\left(u_t+\frac{\xi}{t}\right)-\Lambda_t\left(u_t\right).
$$
From Lemma~\ref{lem:LambdaExpansion} and the fact that $\overline{\Lambda}'_t\left(u_t\right)=0$ 
(see~\eqref{eq:LambdaBar}),
a Taylor expansion around $1$ gives 
\begin{equation}\label{eq:Expansion}
\Lambda_-'\left(1\right)+\left(u_t-1\right)\Lambda_-''\left(1\right)+\frac{2b}{\alpha t\left(1-u_t\right)}+\mathcal{O}\left(\left(1-u_t\right)^2\right)=0.
\end{equation}
From~\eqref{eq:Expansion} we have 
\begin{equation}\label{eq:LimUt}
\lim_{t\to\infty}t\left(1-u_t\right)=-\frac{2b}{\alpha\Lambda'_-\left(1\right)},
\end{equation}
and hence using Lemma~\ref{lem:LambdaExpansion},
$$
\Lambda_t\left(u_t+\frac{\xi}{t}\right)-\Lambda_t\left(u_t\right)
=t\left(\Lambda\left(u_t+\frac{\xi}{t}\right)-\Lambda\left(u_t\right)\right)
-\frac{2b}{\alpha}\log\left(1-\frac{t^{-1}\xi}{1-u_t}\right)
+R_t\left(u_t+\frac{\xi}{t}\right)-R_t\left(u_t\right).
$$
Therefore the equality~\eqref{eq:Expansion} and the limit in~\eqref{eq:LimUt} imply the following behaviours
$$
t\left(\Lambda\left(u_t+\frac{\xi}{t}\right)-\Lambda\left(u_t\right)\right)
=\xi\Lambda'_-\left(1\right)+o\left(1\right)
,\quad\text{and}\quad
-\frac{2b}{\alpha}\log\left(1-\frac{t^{-1}\xi}{1-u_t}\right)
=-\frac{2b}{\alpha}\log\left(1+\frac{\alpha\xi}{2b}\Lambda'_-\left(1\right)\right)+o\left(1\right).$$

Since the function $R_t$ converges on any compact subset of $\mathcal{D}_t$ we eventually obtain
$$\lim_{t\to\infty}\mathbb{E}_{\overline{\mathbb{P}}_t}\left(\exp\left(\xi X_t/t\right)\right)
=\xi\Lambda'_-\left(1\right)-\frac{2b}{\alpha}\log\left(1+\frac{\alpha\xi}{2b}\Lambda'_-\left(1\right)\right),$$
which proves the lemma in the case $\Lambda'_-\left(1\right)<0$.

Let us now consider the case $\Lambda'_-\left(1\right)=0$.
For all $t>0$ and all $\xi$ such that $u_t+\xi/\sqrt{t}\in\mathcal{D}_t$, we have
$$
\log\mathbb{E}_{\overline{\mathbb{P}}_t}\left(\exp\left(\xi\frac{X_t}{\sqrt{t}}\right)\right)
=\log\mathbb{E}_{\overline{\mathbb{P}}_t}\left(\exp\left(\left(u_t+\frac{\xi}{\sqrt{t}}\right)X_t-\Lambda_t\left(u_t\right)\right)\right)
=\Lambda_t\left(u_t+\frac{\xi}{\sqrt{t}}\right)-\Lambda_t\left(u_t\right).
$$
The expansion~\eqref{eq:Expansion} now gives us
\begin{equation}\label{eq:LimUt2}
\lim_{t\to\infty}t\left(1-u_t\right)^2=\frac{2b}{\alpha\Lambda''_-\left(1\right)},
\end{equation}
and hence using Lemma~\ref{lem:LambdaExpansion} we again have
$$
\Lambda_t\left(u_t+\frac{\xi}{\sqrt{t}}\right)-\Lambda_t\left(u_t\right)
=t\left(\Lambda\left(u_t+\frac{\xi}{\sqrt{t}}\right)-\Lambda\left(u_t\right)\right)
-\frac{2b}{\alpha}\log\left(1-\frac{t^{-1/2}\xi}{1-u_t}\right)
+R_t\left(u_t+\frac{\xi}{\sqrt{t}}\right)-R_t\left(u_t\right).
$$
Therefore the equality~\eqref{eq:Expansion} and the limit in~\eqref{eq:LimUt2} imply the following behaviours
\begin{align*}
t\left(\Lambda\left(u_t+\frac{\xi}{\sqrt{t}}\right)-\Lambda\left(u_t\right)\right)
 & = \frac{\xi^2}{2}\Lambda''_-\left(1\right)-\xi\sqrt{\frac{2b\Lambda''_-\left(1\right)}{\alpha}}+o\left(1\right),\\
-\frac{2b}{\alpha}\log\left(1-\frac{t^{-1/2}\xi}{1-u_t}\right)
 & = -\frac{2b}{\alpha}\log\left(1-\xi\sqrt{\frac{\alpha\Lambda''_-\left(1\right)}{2b}}\right)+o\left(1\right).
\end{align*}
Since the function $R_t$ converges on any compact subset of $\mathcal{D}_t$ we obtain
$$\lim_{t\to\infty}\mathbb{E}_{\overline{\mathbb{P}}_t}\left(\exp\left(\xi X_t/\sqrt{t}\right)\right)
=
\frac{\xi^2}{2}\Lambda''_-\left(1\right)-\xi\sqrt{\frac{2b\Lambda''_-\left(1\right)}{\alpha}}
 -\frac{2b}{\alpha}\log\left(1-\xi\sqrt{\frac{\alpha\Lambda''_-\left(1\right)}{2b}}\right)
,$$
which proves the lemma in the case $\Lambda'_-\left(1\right)=0$.
\end{proof}

Following similar steps as in the proof of Proposition~\ref{prop:NonSteep1}, 
Lemma~\ref{lem:LambdaExpansionAt0} above implies that for any $t>0$, the function
$\overline{\Lambda}^{(o)}$ defined on $\mathcal{D}_t\cap\left(0,\infty\right)$ defined by
$$\overline{\Lambda}^{(o)}(u):=\Lambda(u) t-\frac{2b}{\alpha}\log(u),$$
is well defined, convex and steep at the origin.
Furthermore for any $t>0$, there exists a unique $u_t^{(o)}>0$ satisfying 
$\partial_u\overline{\Lambda}^{(o)}\left(u_t^{(o)}\right)=0$
and $u_t^{(o)}$ converges to zero from above as $t$ tends to infinity.
Similar to~\eqref{eq:ProbaPBar}, we can now define a new probability measure $\mathbb{P}^{(o)}$ by
\begin{equation}\label{eq:ProbaPBar0}
\frac{\D\mathbb{P}^{(o)}_t}{\D\mathbb{P}_t}\left(z\right)
:=\exp\Big(u_t^{(o)} z t-\Lambda_t\left(u_t^{(o)}\right)\Big),
\quad\text{for any }z\in\mathbb{R}.
\end{equation}
An analogue to Lemma~\ref{lem:WeakConvergenceAt1}---the proof of which follows similarly---brings the following weak convergence 
result under the probability measure $\mathbb{P}^{(o)}$.

\begin{lemma}\label{lem:WeakConvergenceAt0}
Under the measure $\mathbb{P}^{(o)}_t$ defined in~\eqref{eq:ProbaPBar0}, the sequence of random variables 
$\left(\pi_t^{(1)} X_t/t\right)_t>0$ 
converges weakly to the random variable $Y$ where
$$
Y\equalDistrib
\left\{
\begin{array}{lll}
\displaystyle 
\delta\left(\Lambda'_+\left(0\right)\right)
\ast\Gamma\left(\frac{2b}{\alpha},-\frac{2b}{\alpha\Lambda'_+\left(0\right)}\right),
\quad & \text{and}\quad \pi_t^{(1)}=1,
\quad & \text{if }\Lambda'_+\left(0\right)>0;\\
\\
\displaystyle 
\mathcal{N}\left(-\sqrt{\frac{2b\Lambda''_+\left(0\right)}{\alpha}},\Lambda''_+\left(0\right)\right)
\ast\Gamma\left(\frac{2b}{\alpha},\sqrt{\frac{2b}{\alpha\Lambda''_+\left(0\right)}}\right),
\quad & \text{and}\quad \pi_t^{(1)}=\sqrt{t},
\quad & \text{if }\Lambda'_+\left(0\right)=0.\\
\end{array}
\right.
$$
\end{lemma}


\end{document}